\documentclass[a4paper,UKenglish,cleveref, autoref, thm-restate]{lipics-v2021}

\pdfoutput=1 
\hideLIPIcs  

\usepackage{tabularx,booktabs,multirow} 
\usepackage{scrextend} 

\usepackage{todonotes}

\usepackage{bm} 
\usepackage{setspace}

\usepackage{tikz}
\usetikzlibrary{arrows,decorations.pathmorphing,decorations.pathreplacing,backgrounds,positioning,fit,matrix}
\usetikzlibrary{shapes,calc}

\usepackage{algorithm}
\usepackage{algpseudocode}
\algnewcommand\algorithmicinput{\textbf{Input:}}
\algnewcommand\algorithmicoutput{\textbf{Output:}}
\algnewcommand\Input{\item[\algorithmicinput]}
\algnewcommand\Output{\item[\algorithmicoutput]}

\newcommand{\NN}{\mathbb{N}} 
\newcommand{\RR}{\mathbb{R}} 

\newcommand{\bigO}{\mathcal{O}} 

\newcommand{\FPT}{\textnormal{\textsf{FPT}}}
\newcommand{\Wone}{\textnormal{\textsf{W[1]}}}

\newcommand{\XP}{\textnormal{\textsf{XP}}}
\newcommand{\APX}{\textnormal{\textsf{APX}}}
\newcommand{\NP}{\textnormal{\textsf{NP}}}

\newcommand{\TE}{\mathcal{E}}
\newcommand{\TG}{\mathcal{G}}

\usepackage{xspace}

\newcommand{\N}{\mathbb N}

\newcommand{\TIS}{\textsc{Temporal~$\Delta$ Independent Set}\xspace}
\newcommand{\TC}{\textsc{Temporal~$\Delta$ Clique}\xspace}
\newcommand{\tis}{\textsc{Temporal~$\Delta$ Independent Set}\xspace}

\newcommand{\MTIS}{\textsc{Maximum Temporal~$\Delta$ Independent Set}\xspace}

\newcommand{\is}{\textsc{Independent Set}\xspace}
\newcommand{\OPVD}{vertex deletion set for order preservation\xspace}
\newcommand{\opvd}{OPVD\xspace}

\newcommand{\totalv}{$<_V$\xspace}
\newcommand{\npcomp}{\NP-complete\xspace}
\newcommand{\whard}{\Wone-hard\xspace}
\newcommand{\fpt}{\FPT\xspace}
\newcommand{\wrt}{with respect to\xspace}

\newcommand{\problemdef}[3]{
  \begin{center}
    \begin{minipage}{0.95\textwidth}
      \noindent
      \textsc{#1}

      \vspace{2pt}
      \setlength{\tabcolsep}{3pt}
      \begin{tabularx}{\textwidth}{@{}lX@{}}
        \textbf{Input:} 		& #2 \\
        \textbf{Question:} 	& #3
      \end{tabularx}
    \end{minipage}
  \end{center}
}

\usepackage{etoolbox}
\newcommand{\appsymb}{$\star$}
\newcommand{\appref}[1]{{\hyperref[proof:#1]{\appsymb}}}
\newcommand{\appLink}[1]{{\hyperref[#1]{\appsymb}}}

\newcommand{\commentout}[1]{}

\title{Temporal Interval Cliques and Independent Sets} 

\author{Danny~Hermelin}{Department of Industrial Engineering and Management, Ben-Gurion~University~of~the~Negev,
Beer-Sheva,
Israel}{hermelin@bgu.ac.il}{}{}

\author{Yuval~Itzhaki}{Department of Industrial Engineering and Management, Ben-Gurion~University~of~the~Negev,
Beer-Sheva,
Israel}{itzhaki@campus.tu-berlin.de}{}{}

\author{Hendrik~Molter}{Department of Industrial Engineering and Management, Ben-Gurion~University~of~the~Negev,
Beer-Sheva,
Israel}{molterh@post.bgu.ac.il}{}{}

\author{Rolf~Niedermeier}{TU Berlin, Faculty IV, Algorithmics and Computational Complexity, Germany}{rolf.niedermeier@tu-berlin.de}{}{}

\authorrunning{D.\ Hermelin, Y.\ Itzhaki, H.\ Molter, and R.\ Niedermeier} 

\Copyright{Danny Hermelin, Yuval Itzhaki, Hendrik Molter, and Rolf Niedermeier} 



\category{} 

\relatedversion{} 

\funding{D.\ Hermelin, Y.\ Itzhaki, and H.\ Molter are supported by the ISF, grant No.~1070/20.}

\acknowledgements{The authors want to thank Leon Kellerhals and Ond{\v{r}}ej Such{\'y} for fruitful discussions that led to some of the results for temporal cliques. The authors want to thank anonymous SAND reviewers for their constructive comments.}

\nolinenumbers 

\EventEditors{James Aspnes and Othon Michail}
\EventNoEds{2}
\EventLongTitle{1st Symposium on Algorithmic Foundations of Dynamic Networks (SAND 2022)}
\EventShortTitle{SAND 2022}
\EventAcronym{SAND}
\EventYear{2022}
\EventDate{March 28--30, 2022}
\EventLocation{Virtual Conference}
\EventLogo{}
\SeriesVolume{221}
\ArticleNo{5}

\begin{document}

\maketitle

\begin{abstract}
Temporal graphs have been recently introduced to model changes to a given network that occur throughout a fixed period of time. The \TC problem, that generalizes the well known \textsc{Clique} problem to temporal graphs, has been studied in the context of finding nodes of interest in dynamic networks [TCS '16].
We introduce the \TIS problem, a temporal generalization of \textsc{Independent Set}.
This problem is \emph{e.g.}\ motivated in the context of finding conflict-free schedules for maximum subsets of tasks, that have certain (changing) constraints on each day they need to be performed.
We are specifically interested in the case where each task needs to be performed in a certain time-interval on each day and two tasks are in conflict on a certain day if their time-intervals on that day overlap.
This leads us to considering both problems on the restricted class of temporal unit interval graphs, \emph{i.e.},~temporal graphs where each layer is a unit interval graph.

We present several hardness results as well as positive results. On the algorithmic side, we provide constant-factor approximation algorithms for instances of both problems where $\tau$, the total number of time steps (layers) of the temporal graph, and $\Delta$, a parameter that allows us to model  conflict tolerance, are constants.
We develop an exact \FPT\ algorithm for \TC with respect to parameter $\tau+k$.
Finally, we use the notion of order preservation for temporal unit interval graphs that, informally, requires the intervals of every layer to obey a common ordering. For both problems we provide an \FPT\ algorithm parameterized by the size of minimum vertex deletion set to order preservation.
\end{abstract}

\section{Introduction}
\label{sec:intro}

The analysis of contact patterns between individuals in day-to-day life contexts can deliver great value in research areas such as social sciences or epidemiology of infectious diseases. Various studies have used wearable sensors to record human interaction data in high-schools and hospitals~\cite{fournet2014contact,vanhems2013estimating}.
These records typically contain data in the form of a stream, a series of discrete time steps for each participant, each specifying a set of other participants with whom they interacted. The collection of streams of all participants can be regarded as a temporal binary relation, one that specifies for every moment whether any two participants were in each other's proximity.
Naturally this can be modelled as a \emph{temporal} graph~\cite{viard2016computing}.

Temporal graphs generalize static graphs by adding a discrete temporal dimension to their edge set. Formally, a temporal graph $\TG = (V,\TE,\tau)$ is an ordered triple consisting of a set~$V$ of vertices, a set~$\TE \subseteq \binom{V}{2} \times \{1,2,\dots, \tau\}$ of \emph{time-edges}, and a maximal time label~${\tau \in \N}$. A temporal graph can be regarded as a set of~$\tau$ consecutive \emph{time steps}, in which each step is a static graph.
For~${t \in \{1,\ldots,\tau\}}$, we define the $t$-th layer as $G_t=(V,E_t)$, where~${E_t = \{\{u,v\} : (\{u,v\},t) \in \TE \}}$. We refer to Casteigts \emph{et al.}~\cite{casteigts2012time}, Flocchini \emph{et al.}~\cite{flocchini2013exploration}, Kostakos~\cite{kostakos2009temporal}, Latapy \emph{et al.}~\cite{LVM18} and Michail~\cite{michail2016introduction} for a more detailed background on temporal graphs.

\subparagraph{Temporal cliques.}

In the analysis of the contact patterns between humans, it is natural to look for important groups of people that commonly interact with each other.
When a group of participants come together for a continuous time interval, it is fair to assume that during this time, these have participated in a discussion or a meeting.
Viard et al. introduced the notion of $\Delta$-cliques to find such find events and groups
\cite{banerjee2019enumeration,bentert2019listing,himmel2017adapting,6849333,viard2016computing}.
Given a temporal graph $\TG = (V,\TE,\tau)$ and an integer $\Delta$, we say that a vertex set~$V'\subseteq V$ is a $\Delta$-clique in $\TG$ if it is a clique in the \textit{edge-union} graph of every $\Delta$ consecutive time steps of $\TG$.
That is, for any pair of distinct vertices $v\neq u \in V'$ and~$t \in \{1,..,\tau-\Delta+1\}$, there exists a $t' \in \{t,..,t+\Delta-1\}$ such that $\{u,v\} \in E_{t'}$.
We call this intersection graph of all $\Delta$ consecutive edge-union graphs $$
G=(V,\bigcap \limits_{i=1}^{\tau - \Delta+1} \ \bigcup \limits_{j=i}^{i+\Delta-1} E_j)
$$ 
the \textit{$\Delta$-association} graph of $\TG$. With this notion in mind, we can now define the problem of \textsc{Temporal $\Delta$ Clique} (see~\cref{fig:deltaclique}).

\problemdef{\TC}
{A temporal graph $\TG=(V,\TE,\tau)$ and an integer $k \in \NN$.}
{Is there set $V' \subseteq V$ of vertices such that $|V'| \geq k$ and $V'$ is a clique in the $\Delta$-association graph $G$ of $\TG$?
}

\begin{figure}[t]
\centering

\begin{tikzpicture}

    \node (0) at (4,2.1) {Temporal interval graph $\TG$};
    \draw[rounded corners] (-2.5,2.5) rectangle  (10.5,-2.0) {};
    \node (1) at (0,0) {
    \begin{tikzpicture}[ main/.style = {draw,fill=white, circle,minimum size=0.1cm}]
\def\scale{0.5}
\node (0) at (0,4*\scale) {$G_1$};

  \node[main] (1) at (3*\scale,4*\scale) {$v_1$};
  \node[main] (2) at (5.5*\scale,2*\scale) {$v_2$};
  \node[main] (3) at (4.5*\scale,-0.5*\scale) {$v_3$};
  \node[main] (4) at (1.5*\scale,-0.5*\scale) {$v_4$};
  \node[main] (5) at (0.5*\scale,2*\scale) {$v_5$};

  \draw[thick,-] (1) to (2);
  \draw[thick,-] (1) to (5);
  \draw[thick,-] (3) to (5);
  \draw[thick,-] (4) to (5);

\end{tikzpicture}
    };
    \node (2) at (4,0) {
    \begin{tikzpicture}[main/.style = {draw,fill=white, circle,minimum size=0.1cm}]
\def\scale{0.5}
\node (0) at (0,4*\scale) {$G_2$};
  \node[main] (1) at (3*\scale,4*\scale) {$v_1$};
  \node[main] (2) at (5.5*\scale,2*\scale) {$v_2$};
  \node[main] (3) at (4.5*\scale,-0.5*\scale) {$v_3$};
  \node[main] (4) at (1.5*\scale,-0.5*\scale) {$v_4$};
  \node[main] (5) at (0.5*\scale,2*\scale) {$v_5$};

  \draw[thick,-] (1) to (2);
  \draw[thick,-] (1) to (4);
  \draw[thick,-] (2) to (5);
  \draw[thick,-] (3) to (5);
\end{tikzpicture}
    };
    \node (3) at (8,0) {

\begin{tikzpicture}
[main/.style = {draw,fill=white, circle,minimum size=0.1cm}]
\def\scale{0.5}
\node (0) at (0,4*\scale) {$G_3$};

  \node[main] (1) at (3*\scale,4*\scale) {$v_1$};
  \node[main] (2) at (5.5*\scale,2*\scale) {$v_2$};
  \node[main] (3) at (4.5*\scale,-0.5*\scale) {$v_3$};
  \node[main] (4) at (1.5*\scale,-0.5*\scale) {$v_4$};
  \node[main] (5) at (0.5*\scale,2*\scale) {$v_5$};

  \draw[thick,-] (1) to (3);
  \draw[thick,-] (1) to (4);
  \draw[thick,-] (1) to (5);
  \draw[thick,-] (2) to (5);

\end{tikzpicture}

};
    \node (4) at (2,-4.25) {

\begin{tikzpicture}
[main/.style = {draw,fill=white, circle,minimum size=0.1cm}]
\def\scale{0.5}
\node (0) at (1.5,4*\scale-3){$\Delta$-conflict graph};

  \node[main] (1) at (3*\scale,4*\scale) {$v_1$};
  \node[main] (2) at (5.5*\scale,2*\scale) {$v_2$};
  \node[main] (3) at (4.5*\scale,-0.5*\scale) {$v_3$};
  \node[main] (4) at (1.5*\scale,-0.5*\scale) {$v_4$};
  \node[main] (5) at (0.5*\scale,2*\scale) {$v_5$};
    
  \draw[thick,-] (1) to (4);
  \draw[thick,-] (1) to (2);
  \draw[thick,-] (2) to (5);
  \draw[thick,-] (3) to (5);

\end{tikzpicture}

};

    \node (5) at (6,-4.25) {

\begin{tikzpicture}
[main/.style = {draw,fill=white, circle,minimum size=0.1cm}]
\def\scale{0.5}
\node (0) at (1.5,4*\scale-3) {$\Delta$-association graph};

  \node[main] (1) at (3*\scale,4*\scale) {$v_1$};
  \node[main] (2) at (5.5*\scale,2*\scale) {$v_2$};
  \node[main] (3) at (4.5*\scale,-0.5*\scale) {$v_3$};
  \node[main] (4) at (1.5*\scale,-0.5*\scale) {$v_4$};
  \node[main] (5) at (0.5*\scale,2*\scale) {$v_5$};

  \draw[thick,-] (1) to (2);
  \draw[thick,-] (1) to (4);
  \draw[thick,-] (1) to (5);
  \draw[thick,-] (2) to (5);
  \draw[thick,-] (3) to (5);

\end{tikzpicture}

};

\end{tikzpicture}

\caption{An example a temporal graph~$\TG$ with three layers, along with its $\Delta$-association graph and $\Delta$-conflict graph, for $\Delta=2$. The vertex subset~$\{v_1, v_2, v_5\}$ is a maximum sized clique for this instance, while the subset~$\{v_2, v_3, v_4\}$ is a maximum sized independent set.}
\label{fig:deltaclique}
\end{figure}
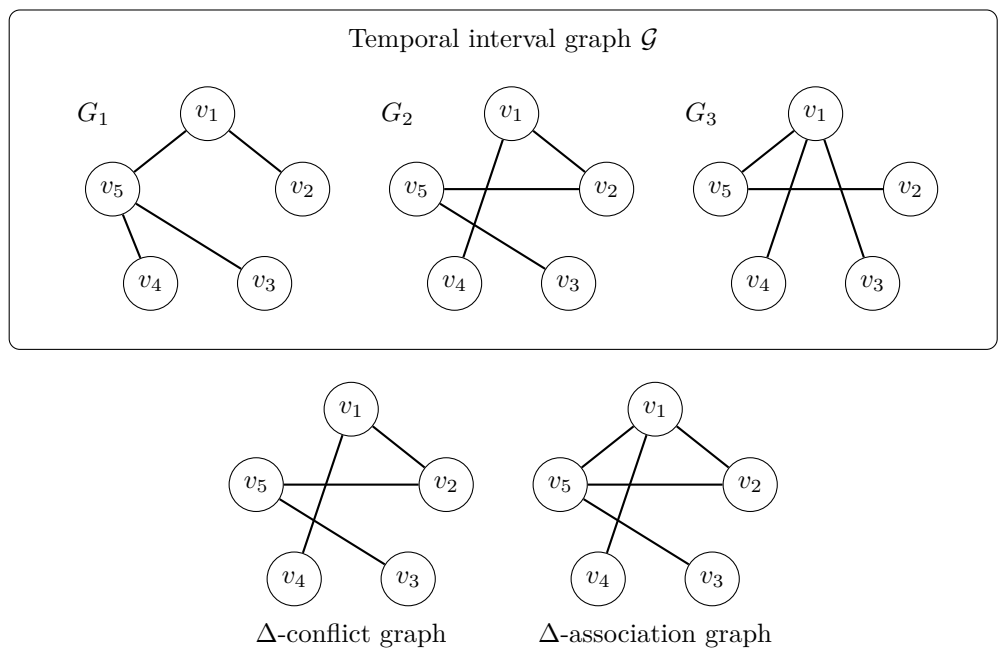

\subparagraph{Temporal independent set.}

The natural ``dual'' of \textsc{Temporal~$\Delta$ Clique} is the \textsc{Temporal~$\Delta$ Independent Set} problem. In the scenarios consider above, one might look for a group of people which are mutually non-interactive. This is natural when one wishes to schedule shared resources such as school lab or hospital ward among different human participants.

Given a temporal graph $\TG=(V,\TE,\tau)$, we say a vertex set~$V'$ is a  \emph{$\Delta$-independent set} in $\TG$ if $V'$ is an independent set in the edge-\emph{intersection} graph of every~$\Delta$ consecutive time steps of $\TG$. That is, for any pair of distinct vertices~$u \neq v \in V'$ and~$t \in \{1,\ldots,\tau-\Delta+1\}$, there exists a~$t' \in \{t,\ldots,t+\Delta-1\}$ such that~$\{u,v\}\notin E_{t'}$. We call this edge-intersection graph of every~$\Delta$ consecutive time steps 
$$
G=(V,\bigcup \limits_{i=1}^{\tau - \Delta+1} \ \bigcap \limits_{j=i}^{i+\Delta-1} E_j)
$$ 
the \emph{$\Delta$-conflict graph} of $\TG$. With this notion in mind, we can now introduce the second problem we deal with in this paper (see~\cref{fig:deltaclique}).

\bigskip

\problemdef{\TIS}
{A temporal graph $\TG=(V,\TE,\tau)$ and an integer $k \in \NN$.}
{Is there set $V' \subseteq V$ of vertices such that $|V'| \geq k$ and $V'$ is an independent set in the $\Delta$-conflict graph $G$ of $\TG$?
}

\subparagraph{Temporal interval graphs.}

Recall our problem of analysing contact patterns between human participants. Observe that in this setting, each daily conflict graph $G_t$ (corresponding to layer $t$ of the input temporal graph) can also be represented by a set of $n$ intervals, where each each interval indicates the time-interval of the student attendance in the given study room.
As first defined by Haj\'os~\cite{hajos1957uber}, a graph belongs to the class of interval graphs if there exists a mapping of its vertices to a set of intervals over a line such that two vertices are adjacent if and only if their corresponding intervals overlap. Interval graphs are used to model many natural phenomena which occur along the line of a one-dimensional axis, and have various applications in scheduling~\cite{bar2006scheduling}, computational biology~\cite{joseph1992determining}, and many other areas.

An important subclass of interval graphs is the class of \emph{unit interval graphs}: A graph~$G$ is a \emph{unit interval graph} if it has an interval representation where all intervals are of the same length. It is well-known that this graph class is equivalent to the class of \emph{proper interval graphs}, graphs with interval representation where no interval is properly contained in another~\cite{roberts1969indifference}. The restriction to unit interval graphs is quite natural for our contact pattern analysis problem, since in many cases one can assume that all participants attend their meeting place for roughly the same time. We will therefore mostly focus on the \TC and \TIS problems restricted to temporal unit interval graphs, that is temporal graphs where each layer is a unit interval graph. 

\subparagraph{Order-preserving temporal interval graphs.}
Considering our example of \TC, it may be a reasonable assumption that some participants generally prefer to meet in the morning while others prefer to meet in the evenings. In this scenario, we have a natural ordering on the time-intervals of the participants that stays the same or at least does not change much over the time period of $\tau$ days. We use the notion of \emph{order-preserving temporal graphs} to formalize this setting.

Order preservation on temporal interval graphs was first introduced by Fluschnik \emph{et al.}~\cite{Nie2020temporal}. A \emph{temporal interval graph} is said to be \emph{order-preserving} if it admits a vertex ordering~$<_V$ such that each of its time steps can be represented by an interval model such that \emph{both} the right-endpoints and left-endpoints are ordered by~$<_V$. Fluschnik \emph{et al.}~\cite{Nie2020temporal} show that the recognition of order-preserving temporal unit interval graphs can be done in linear time, and also offer a metric to measure the distance of a temporal interval graph from being order-preserving, which they call the ``shuffle number''. It measures the maximum pairwise disagreements in the vertex ordering of any two consecutive layers. We propose an alternative metric to measure the distance to order preservation. Our distance is simply the minimum number $\ell$ of vertices to be deleted in order to obtain an order-preserving temporal interval graph, and we call it the order-preserving vertex deletion (OPVD) metric.

\subsection{Our results}

We present both positive and negative results regarding \TC and \TIS in temporal (unit) interval graphs. 

We begin with \TC in Section~\ref{sec:clique}, by first reviewing previously known results that carry over from $\tau$-track interval graphs. We then present an approximation algorithm for the unit interval case with a factor of~$(2\Delta)^{\tau-\Delta+1}$, and an \FPT~algorithm for this case with respect to the parameter $\tau + k$. As the problem is known from previous results to be polynomial-time solvable for $\Delta  = \tau = 2$, and \NP-hard for $\Delta = \tau = 3$, we complement these two results with an \NP-hardness result for the case where $\Delta=2$ and~$\tau=11$. We also give a \Wone-hardness proof for parameter $k + \tau$, when the intervals are not necessarily of unit length.

We proceed to study \TIS in Section~\ref{sec:IS}. We show that the simple greedy algorithm achieves an approximation factor of $2^\Delta(\tau-\Delta+1)$ for the unit interval case. We then show that previous results on $\tau$-track interval graphs limit any further positive results, as the problem is known to be both \APX-hard and \Wone-hard even when there are only two time steps (\emph{i.e.} $\tau =2$).

As both problems are already hard in quite restrictive settings, we turn to discuss order-preserving temporal unit interval graphs in~\cref{chap:opTG}. We show that computing the OPVD set (\emph{i.e.},~the set of vertices whose removal leaves an order-preserving graph) of a unit interval temporal graph is \NP-hard. We complement this result by providing an \fpt algorithm for computing an OPVD set when parameterized by the solution size. This leads to an \fpt algorithm for both \TIS and \TC on temporal unit interval graphs when parameterized by minimum OPVD set.



\subsection{Related Work}

By now, there is already a significant body of research related to temporal graphs in general~\cite{casteigts2012time,flocchini2013exploration,kostakos2009temporal,michail2016introduction,hamm2022complexity}, as well as graph problems cast onto the temporal setting~\cite{akrida2020temporal,bentert2019listing,Nie2020temporal,himmel2017adapting,mertzios2019computing,mertzios2019sliding,viard2016computing}. There also has been previous work considering special temporal graph classes, mostly for the temporal separator problem~\cite{Nie2020temporal,MaackMNR21}.

To the best of our knowledge, the problem of \TIS has not been studied previously, but our definition is highly inspired by the \TC problem~\cite{bentert2019listing,himmel2017adapting,viard2016computing}. Viard \emph{et al.}~\cite{viard2016computing} give an exponential-time algorithm for \TC. Himmel \emph{et al.}~\cite{himmel2017adapting} have shown that \TC is in \fpt when parameterized by the so-called \emph{$\Delta$-slice degeneracy}. Bentert \emph{et al.}~\cite{bentert2019listing} generalized the result by Himmel \emph{et al.}~\cite{himmel2017adapting} for temporal $s$-plexes, a generalization of temporal cliques.

The classical static problems \is and \textsc{Clique} are clearly special cases of \tis and \TC when $\tau=1$. While both are \NP-complete for general undirected graphs \cite{FrancisGO15,GJ79}, both are solvable in polynomial time on interval graphs and some of their generalizations~\cite{da2007triangulated,FrancisGO15,gavril1974intersection,hsu1995independent,DBLP:journals/siamcomp/RoseTL76}. Thus, both \tis and \TC on temporal interval graphs are polynomial time solvable when $\tau=1$.

Moreover, for arbitrary values of $\tau$, the \textsc{Temporal~1 Independent Set} and \textsc{Temporal~$\tau$ Clique} problem are special cases of \is and \text{Clique} on $\tau$-track interval graphs~\cite{gyarfas1995multitrack}. Bar-Yehuda \emph{et al.}~\cite{bar2006scheduling} presented a $2\tau$ approximation algorithm for \is in $\tau$-track graphs, while Fellows \emph{et al.}~\cite{fellows2009multipleinterval} and Jiang~\cite{jiang2010parameterized} studied this problem from the perspective of parameterized complexity. K{\"o}nig~\cite{konig2010} showed that \textsc{Clique} is polynomial-time solvable on $2$-track interval graphs. Francis \emph{et al.}~ \cite{FrancisGO15} showed \NP-hardness on~$3$-track unit interval graphs.
as well as \APX-hardness on $\tau$-track graphs. Butman \emph{et al.}~\cite{butman2010optimization} presented a $(\tau^2-\tau+1)/2$-approximation for its containing graph class of $\tau$-interval graphs. However, they do not rule out the existence of a constant factor approximation.

The problems \textsc{Temporal $\tau$ Independent Set} and \textsc{Temporal 1 Clique} are special cases of \is and \textsc{Clique} on intersection graphs of $\tau$-dimensional hyperrectangles. Marx~\cite{marx2005efficient} showed that \textsc{Independent Set} is \npcomp~and \whard~with respect to the solution size when restricted to the intersection graphs of axis-parallel unit squares in the plane. Chleb{\'i}k and Chleb{\'i}kov{\'a}~\cite{chlebik2005approximation} proved, for instance, that \textsc{Maximum Independent Set} is \APX-hard for intersection graphs of $d$-dimensional \emph{rectangles}, yet on such graphs the optimal solution can be approximated within a factor of $d$~\cite{akcoglu2002opportunity}. On intersection graphs of~$d$-dimensional \emph{squares} \textsc{Maximum Independent Set} admits a \emph{polynomial time approximation scheme (PTAS)} for a constant~$d$~\cite{chan2003polynomial,henzinger2020dynamic,khanna1998approximating}.
Rosgen \emph{et al.} \cite{RosgenS07} showed that on hyperrectangles intersection graphs, \text{Clique} has an \XP~algorithm with respect to the dimension~$d$.


\section{Preliminaries}
\label{sec:Preliminaries}%
In this section, we first introduce all temporal graph notation and terminology used in this work, including basic concepts on interval graphs and unit interval graphs. In the final part of the section we also discuss a geometric representation of a $\Delta$-association and a $\Delta$-conflict graph of a given temporal interval graph. 

\subsection{Basic notation and definitions.}

Let $a,b \in \NN$ such that~${a<b}$. We denote the set of all integers $x$ with $a \leq x \leq b$ by $[a:b]$. As a shorthand, we use~$[b]$ when $a=1$. We also~$[a,b] \subseteq \RR$ to denote the set of real numbers between $a$ and $b$.

Let~$G=(V,E)$ denote an undirected graph, where~$V$ denotes the set of vertices and~$E\subseteq \{\{v,w\}\mid v,w\in V,\, v\neq w\}$ denotes the set of edges.
For a graph~$G$, we also write~$V(G)$ and~$E(G)$ to denote the set of vertices and the set of edges of~$G$, respectively. We denote~${n:=|V|}$. Given an ordering~\totalv over the vertices $V$ of a graph $G=(V,E)$ in which~$v_i$ is the $i$-th vertex in the ordering, we denote by $V_{[a:b]}$ the set $\{ v_i \mid i \in [a:b] \}$ and by~$G_{[a:b]}$ the graph induced by~$V_{[a:b]}$. We use the notation $\text{index}_{<_V}(v)$ for the ordinal position of $v$ in~\totalv.

An undirected \emph{temporal graph}~$\TG = (V,\TE,\tau)$ is an ordered triple consisting of a set~$V$ of vertices, a set ~$\TE \subseteq \binom{V}{2} \times [\tau]$ of \emph{time-edges}, and a maximal time label~$\tau \in \N$. Given a temporal graph~$\TG = (V,\TE,\tau)$, we denote by $E_t$ the set of all edges that are available at time $t$, that is, $E_t:=\{\{v,w\}\mid (\{v,w\},t)\in\TE\}$ and by $G_t$ the $t$-th \emph{layer} of $\TG$, that is, $G_t:=(V,E_t)$.
For two graphs $G_1$ and $G_2$ over the same vertex set $V$, we denote by:
\begin{itemize}
	\item $G_1 \cap G_2$ the \emph{edge-intersection graph} of~$G_1$ and $G_2$, formally $G_1 \cap G_2 := (V, E_1\cap E_2 )$,
  \item $G_1 \cup G_2$ the \emph{edge-union graph} of~$G_1$ and $G_2$, formally $G_1 \cup G_2 := (V, E_1\cup E_2 )$, and
  \item $\TG-V'$ the \emph{temporal graph} induced by~$V\setminus V'$, formally $\TG-V':=(V \setminus V', \TE',\tau)$ with~$\TE'=\{ (\{v,u\}, t) \mid v,u\in V\setminus V' \wedge (\{v,u\}, t) \in \TE  \}$.
\end{itemize}

\subsection{Geometric intersection graphs}

An undirected graph is an \emph{interval} graph if there exists a mapping from its vertices to intervals on the real line so that two vertices are adjacent if and only if their intervals intersect \cite{gilmore1964characterization}. Such a representation is called an \emph{intersection model} or an \emph{interval representation}. Formally, given an interval graph $G=(V,E)$, an \emph{interval representation} for $G$ is a mapping of each vertex $v\in V$ to an interval~$\rho(v) \subset \RR$ such that~${E=\{ \{v,u\} \subseteq V \mid \rho(v) \cap \rho(u) \neq \varnothing \}}$. We denote by $\text{right}_{\rho}(u)$ and by $\text{left}_{\rho}(u)$ the real value of the right and left endpoints of $v$'s associated interval on the interval representation $\rho$; the subscript $\rho$ will be omitted if it is clear from the context to which representation we refer. We let~$\rho(G)$ denote the entire representation of $G$. If the length of all intervals in~$\rho(G)$ are equal, then $\rho(G)$ is a \emph{unit interval representation}, and $G$ is an \emph{unit interval graph}~\cite{roberts1969indifference}.   

An important generalization of (unit) interval graphs in our context is the class of \emph{$t$-track interval graphs}~\cite{gyarfas1995multitrack}, for some integer $t \geq 1$. A graph $G=(V,E)$ is said to be a $t$-track interval graph if there are $\tau$ interval graphs $G_1=(V,E_1),\ldots,G_t=(V,E_t)$ such that $E=\bigcup^d_{i=1} E_i$. A $t$-track interval graph $G$ has a useful geometric representation as well. Formally, a \emph{$t$-track interval representation} for a $t$-track interval graph $G=(V,E)$ is a mapping of each vertex~$v\in V$ to a set of $t$ disjoint intervals~$\rho(v)=\{\rho_1(v),\ldots,\rho_t(v)\}$ such that $\{u,v\} \in E$ iff~$\rho_i(u) \cap \rho_i(v) \neq \emptyset$ for some $i \in \{1,\ldots,t\}$. In this way, one can think of the real line as partitioned into $t$ disjoint segments (tracks) such that the $i$'th interval of all vertices are contained strictly in segment~$i$. If all intervals in $\rho(G)$ are of the same length, we call $G$ a~\emph{$\tau$-track unit interval graph}. 

Another important generalization of (unit) interval graphs in is the class of $d$-dimensional hyperrectangle graphs, for a given integer $d \geq 1$. These are no more than intersection graphs of axis-parallel hyperrectangles in $\RR^d$. A \emph{hyperrectangle representation} of a $d$-dimensional hyperrectangle graph $G=(V,E)$ is a mapping of each vertex $v\in V$ to an axis-parallel hyperrectangle~$\rho(v) \subset \RR^d$ such that~$\{u,v\} \in E$ iff $\rho(u) \cap \rho(v) \neq \emptyset$ for some $i \in \{1,\ldots,t\}$. It is well known that $G$ is a $d$-dimensional hyperrectangle graph iff there exist $d$ interval graphs~$G_1=(V,E_1), \ldots, G_d=(V,E_d)$ such that $E=\bigcap^d_{i=1} E_i$.

\subsection{Geometric interpretation of association and conflict graphs}

We next consider geometric representations of $\Delta$-association and $\Delta$-conflict graphs that will prove useful throughout the paper. Let $\TG = (V,\TE,\tau)$ be a temporal interval graph, and let~$E_1,\ldots,E_\tau$ be the edge sets corresponding to the $\tau$ time-steps of $G$. Moreover, let~$\rho_i=\rho(G_i))$ be the interval representation of $G_i = (V,E_i)$, for each $i \in \{1,\ldots,\tau\}$.

\subparagraph{$\Delta$-association graph:}

Let $G=(V,E)$ be the $\Delta$-association graph of $\TG$. Consider first the case of~$\Delta=1$. In this case we have that $E=\bigcap^\tau_{i=1} E_i$, and so $G$ is formed by taking the intersection of $\tau$ interval graphs. Thus, $G$ is a $\tau$-dimensional hyperrectangle graph by definition, and each vertex $v \in V$ can be represented by the hyperrectangle formed by~$\rho_1(v), \ldots,\rho_\tau(v)$. If $\Delta = \tau$, then $E=\bigcup^\tau_{i=1} E_i$, implying that $G$ is a $\tau$-track interval graph, where in $\rho(G)$, each vertex $v$ is mapped to the interval set $\{\rho_1(v),\ldots,\rho_\tau(v)\}$.  

Next consider the case of $1 < \Delta < \tau$. In this case, we get a hybrid of both cases above, and we need to consider several hyperrectangles that are associated with each vertex. Observe that each ``$\Delta$-window" of~$\TG$ is formed by taking the union of $\Delta$ consecutive time-steps, and so it corresponds to a $\Delta$-track interval graph.
Moreover, the number of $\Delta$-windows in $\TG$ is ${\tau-\Delta+1}$, and $G$ is formed by taking the intersection of these $\Delta$-windows.
Thus, one can think of each $\Delta$-track interval graph as existing on a different axis, and a vertex is now associated with a set of $\Delta^{\tau-\Delta+1}$ hyperrectangles, a hyperrectangle for each combination of one of the $\Delta$ intervals on each of the ${\tau-\Delta+1}$ axis (see~\cref{fig:tsquarerepC}).
The dimension of each of these hyperrectangles is $\tau-\Delta+1$.
An edge exists between two vertices in $G$ if any two of their $(\tau-\Delta+1)$-dimensional hyperrectangles intersect.
Hence, $G$ belongs to the class of~$\Delta^{\tau-\Delta+1}$-track $(\tau-\Delta+1)$-dimensional hyperrectangle graphs.

\begin{corollary}
\label{cor:association}%
The $\Delta$-association graph of any temporal interval graph is a~$\Delta^{\tau-\Delta+1}$-track $(\tau-\Delta+1)$-dimensional hyperrectangle graph.
\end{corollary}

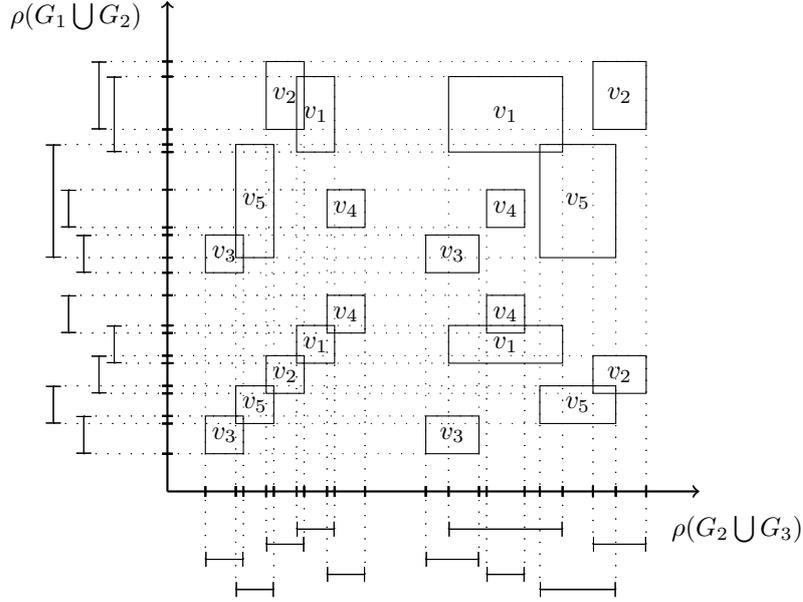
\begin{figure}[h!]
\centering

  \begin{tikzpicture}

        \def\shitw{1}
        \def\shith{0.5}

      \draw[thick,->] (0,0) -- (7.,0) node[anchor=north west] {};
    \draw[thick,->] (0,0) -- (0,6.5) node[anchor=south east] {};
    \draw[] (-1.2,6.3) node[] {$\rho(G_1\bigcup G_2)$};
     \draw[] (7.5,-0.5) node[] {$\rho(G_2 \bigcup G_3)$};

    \foreach \i\x/\w/\y/\h in {
    1/1.2-0.5/0.5/1.2/0.5,
    2/0.8-0.5/0.5/0.8/0.5,
    3/0-0.5/0.5/0/0.5,
    4/1.6-0.5/0.5/1.6/0.5,
    5/0.4-0.5/0.5/0.4/0.5
    }
    {
        \draw[thick] (\shitw+\x ,2pt) -- (\shitw+\x ,-2pt) node[anchor=north] {};
        \draw[thick] (2pt,\shith+\y) -- (-2pt,\shith+\y) node[anchor=east] {};
        \draw[thick] (\shitw+\x+\w , 2pt) -- (\shitw+\x+\w ,-2pt) node[anchor=north] {};
        \draw[thick] (2pt,\shith+ \y+\h) -- (-2pt,\shith+\y+\h ) node[anchor=east] {};
        \draw[loosely dotted] (-\i/5-0.5,\shith+\y) -- (\shitw+\x,\shith+\y) node[anchor=south east] {};
        \draw[loosely dotted] (\shitw+\x,\shith+-\i/5-0.8) -- (\shitw+\x,\shith+\y) node[anchor=south east] {};
        \draw[loosely dotted] (-\i/5-0.5,\shith+\y+\h) -- (\shitw+\x+\w,\shith+\y+\h) node[anchor=south east] {};
        \draw[loosely dotted] (\shitw+\x+\w,\shith+-\i/5-0.8) -- (\shitw+\x+\w,\shith+\y+\h) node[anchor=south east] {};
        \draw (\shitw+\x ,\shith+\y) rectangle (\shitw+\x+\w,\shith+\y+\h) {};
        \draw[] (\shitw+\x+\w/2 ,\shith+\y+\h/2) node[] {$v_{\i}$};

        \draw[|-|] (\shitw+\x, \shith+-\i/5-0.8 ) -- (\shitw+\x+\w, \shith+-\i/5-0.8 ) node[anchor=north] {};
        \draw[|-|] (-\i/5-0.5,\shith+ \y ) -- (-\i/5-0.5, \shith+\y+\h ) node[anchor=north  ] {};

        } {}

    \foreach \i\x/\w/\y/\h in {
        1/-0.3+3*\shitw/1.5/1.2/0.5,
        2/1.6+3*\shitw/0.7/0.8/0.5,
        3/-0.6+3*\shitw/.7/0/0.5,
        4/0.2+3*\shitw/.5/1.6/0.5,
        5/0.9+3*\shitw/1./0.4/0.5
    }
    {
        \draw[thick] (\shitw+\x ,2pt) -- (\shitw+\x ,-2pt) node[anchor=north] {};
        \draw[thick] (2pt,\shith+\y) -- (-2pt,\shith+\y) node[anchor=east] {};
        \draw[thick] (\shitw+\x+\w , 2pt) -- (\shitw+\x+\w ,-2pt) node[anchor=north] {};
        \draw[thick] (2pt,\shith+ \y+\h) -- (-2pt,\shith+\y+\h ) node[anchor=east] {};
        \draw[loosely dotted] (-\i/5-0.5,\shith+\y) -- (\shitw+\x,\shith+\y) node[anchor=south east] {};
        \draw[loosely dotted] (\shitw+\x,\shith+-\i/5-0.8) -- (\shitw+\x,\shith+\y) node[anchor=south east] {};
        \draw[loosely dotted] (-\i/5-0.5,\shith+\y+\h) -- (\shitw+\x+\w,\shith+\y+\h) node[anchor=south east] {};
        \draw[loosely dotted] (\shitw+\x+\w,\shith+-\i/5-0.8) -- (\shitw+\x+\w,\shith+\y+\h) node[anchor=south east] {};
        \draw (\shitw+\x ,\shith+\y) rectangle (\shitw+\x+\w,\shith+\y+\h) {};
        \draw[] (\shitw+\x+\w/2 ,\shith+\y+\h/2) node[] {$v_{\i}$};

        \draw[|-|] (\shitw+\x, \shith+-\i/5-0.8 ) -- (\shitw+\x+\w, \shith+-\i/5-0.8 ) node[anchor=north] {};
        \draw[|-|] (-\i/5-0.5,\shith+ \y ) -- (-\i/5-0.5, \shith+\y+\h ) node[anchor=north  ] {};

        } {}

    \foreach \i\x/\w/\y/\h in {
        1/-0.3+3*\shitw/1.5/1.+3*\shitw/1.,
        2/1.6+3*\shitw/0.7/1.3+3*\shitw/0.9,
        3/-0.6+3*\shitw/.7/-0.6+3*\shitw/0.5,
        4/0.2+3*\shitw/0.5/0+3*\shitw/0.5,
        5/0.9+3*\shitw/1./-0.4+3*\shitw/1.5
    }
    {
        \draw[thick] (\shitw+\x ,2pt) -- (\shitw+\x ,-2pt) node[anchor=north] {};
        \draw[thick] (2pt,\shith+\y) -- (-2pt,\shith+\y) node[anchor=east] {};
        \draw[thick] (\shitw+\x+\w , 2pt) -- (\shitw+\x+\w ,-2pt) node[anchor=north] {};
        \draw[thick] (2pt,\shith+ \y+\h) -- (-2pt,\shith+\y+\h ) node[anchor=east] {};
        \draw[loosely dotted] (-\i/5-0.5,\shith+\y) -- (\shitw+\x,\shith+\y) node[anchor=south east] {};
        \draw[loosely dotted] (\shitw+\x,\shith+-\i/5-0.8) -- (\shitw+\x,\shith+\y) node[anchor=south east] {};
        \draw[loosely dotted] (-\i/5-0.5,\shith+\y+\h) -- (\shitw+\x+\w,\shith+\y+\h) node[anchor=south east] {};
        \draw[loosely dotted] (\shitw+\x+\w,\shith+-\i/5-0.8) -- (\shitw+\x+\w,\shith+\y+\h) node[anchor=south east] {};
        \draw (\shitw+\x ,\shith+\y) rectangle (\shitw+\x+\w,\shith+\y+\h) {};
        \draw[] (\shitw+\x+\w/2 ,\shith+\y+\h/2) node[] {$v_{\i}$};

        \draw[|-|] (\shitw+\x, \shith+-\i/5-0.8 ) -- (\shitw+\x+\w, \shith+-\i/5-0.8 ) node[anchor=north] {};
        \draw[|-|] (-\i/5-0.5,\shith+ \y ) -- (-\i/5-0.5, \shith+\y+\h ) node[anchor=north  ] {};

        } {}

        \foreach \i\x/\w/\y/\h in {
    1/1.2-0.5/0.5/1.+3*\shitw/1.,
    2/0.8-0.5/0.5/1.3+3*\shitw/0.9,
    3/0-0.5/0.5/-0.6+3*\shitw/0.5,
    4/1.6-0.5/0.5/0+3*\shitw/0.5,
    5/0.4-0.5/0.5/-0.4+3*\shitw/1.5
    }
    {
        \draw[thick] (\shitw+\x ,2pt) -- (\shitw+\x ,-2pt) node[anchor=north] {};
        \draw[thick] (2pt,\shith+\y) -- (-2pt,\shith+\y) node[anchor=east] {};
        \draw[thick] (\shitw+\x+\w , 2pt) -- (\shitw+\x+\w ,-2pt) node[anchor=north] {};
        \draw[thick] (2pt,\shith+ \y+\h) -- (-2pt,\shith+\y+\h ) node[anchor=east] {};
        \draw[loosely dotted] (-\i/5-0.5,\shith+\y) -- (\shitw+\x,\shith+\y) node[anchor=south east] {};
        \draw[loosely dotted] (\shitw+\x,\shith+-\i/5-0.8) -- (\shitw+\x,\shith+\y) node[anchor=south east] {};
        \draw[loosely dotted] (-\i/5-0.5,\shith+\y+\h) -- (\shitw+\x+\w,\shith+\y+\h) node[anchor=south east] {};
        \draw[loosely dotted] (\shitw+\x+\w,\shith+-\i/5-0.8) -- (\shitw+\x+\w,\shith+\y+\h) node[anchor=south east] {};
        \draw (\shitw+\x ,\shith+\y) rectangle (\shitw+\x+\w,\shith+\y+\h) {};
        \draw[] (\shitw+\x+\w/2 ,\shith+\y+\h/2) node[] {$v_{\i}$};

        \draw[|-|] (\shitw+\x, \shith+-\i/5-0.8 ) -- (\shitw+\x+\w, \shith+-\i/5-0.8 ) node[anchor=north] {};
        \draw[|-|] (-\i/5-0.5,\shith+ \y ) -- (-\i/5-0.5, \shith+\y+\h ) node[anchor=north  ] {};

        } {}

    \end{tikzpicture}\label{tik1}
\caption{Geometric representation of the $\Delta$-association graph from \cref{fig:deltaclique}. There exist an edge between any $v,u \in V$ if and only if there exist a rectangle of $v$ and a rectangle of $u$ which intersect, $\rho(v)\bigcap \rho(u) \neq \varnothing$.
}
\label{fig:tsquarerepC}
\end{figure}

\subparagraph{$\Delta$-conflict graph:}

The $\Delta$-conflict graph $G=(V,E)$ of $\TG$ has a useful geometric interpretation as well. For the cases of $\Delta =1$ and $\Delta = \tau$, the situation is flipped. When $\Delta = 1$, we have~$E = \bigcup^\tau_{i=1} E_i$, and so $G$ is $\tau$-track interval graph. Indeed, a $\tau$-track interval representation of $G$ can be obtained by taking $\rho(v)$ to be the set of disjoint $\tau$ intervals $\rho_1(v),\ldots,\rho_\tau(v)$. When $\Delta = \tau$, we have~$E = \bigcap^\tau_{i=1} E_i$, and so $G$ is $\tau$-dimensional hyperrectangle graph, where each vertex $v$ can be represented by the $\tau$-dimensional formed from $\rho_1(v),\ldots,\rho_\tau(v)$.

For $1 < \Delta < \tau$, we again get a hybrid of both cases above. In this case, each 
$\Delta$-window of~$\TG$ is formed by taking the intersection of $\Delta$ consecutive time-steps, and so it corresponds to a~$\Delta$-dimensional hyperrectangle graph. Moreover, the number of $\Delta$-windows in $\TG$ is~${\tau-\Delta+1}$, and $G$ is formed by taking the union of these $\Delta$-windows. In this way, $G$ is a~$(\tau-\Delta+1)$-track~$\Delta$-dimensional hyperrectangle graph.
As shown in \cref{fig:tsquarerepIS}, an edge exists between two vertices in $G$ iff any two of their $(\tau-\Delta+1)$-dimensional hyperrectangles intersect in some track. 

       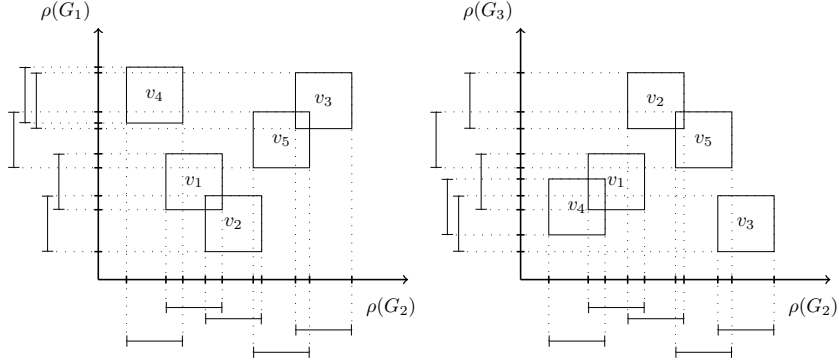
\begin{figure}[t]
\centering

    \resizebox{\textwidth*4/5}{!}{
  \begin{tikzpicture}

        \def\shitw{1}
        \def\shith{0.5}

      \draw[thick,->] (0,0) -- (5.5,0) node[anchor=north west] {};
     \draw[] (5.2,-0.5) node[] {$\rho(G_2)$};
      \draw[thick,->] (0,0) -- (0,4.5) node[anchor=south east] {$\rho(G_1)$};

    \foreach \i\x/\w/\y/\h\c in {
        1/0.2/1/0.75/1/yellow,
        2/0.9/1/0/1/black,
        3/2.5/1/2.2/1/black,
        4/-0.5/1/2.3/1/yellow,
        5/1.75/1/1.5/1/yellow
    }
    {
        \draw[thick] (\shitw+\x ,2pt) -- (\shitw+\x ,-2pt) node[anchor=north] {};
        \draw[thick] (2pt,\shith+\y) -- (-2pt,\shith+\y) node[anchor=east] {};
        \draw[thick] (\shitw+\x+\w , 2pt) -- (\shitw+\x+\w ,-2pt) node[anchor=north] {};
        \draw[thick] (2pt,\shith+ \y+\h) -- (-2pt,\shith+\y+\h ) node[anchor=east] {};
        \draw[loosely dotted] (-\i/5-0.5,\shith+\y) -- (\shitw+\x,\shith+\y) node[anchor=south east] {};
        \draw[loosely dotted] (\shitw+\x,\shith+-\i/5-0.8) -- (\shitw+\x,\shith+\y) node[anchor=south east] {};
        \draw[loosely dotted] (-\i/5-0.5,\shith+\y+\h) -- (\shitw+\x+\w,\shith+\y+\h) node[anchor=south east] {};
        \draw[loosely dotted] (\shitw+\x+\w,\shith+-\i/5-0.8) -- (\shitw+\x+\w,\shith+\y+\h) node[anchor=south east] {};
        \draw (\shitw+\x ,\shith+\y) rectangle (\shitw+\x+\w,\shith+\y+\h) {};
        \draw[] (\shitw+\x+\w/2 ,\shith+\y+\h/2) node[] {$v_{\i}$};

        \draw[|-|] (\shitw+\x, \shith+-\i/5-0.8 ) -- (\shitw+\x+\w, \shith+-\i/5-0.8 ) node[anchor=north] {};
        \draw[|-|] (-\i/5-0.5,\shith+ \y ) -- (-\i/5-0.5, \shith+\y+\h ) node[anchor=north  ] {};

        } {}

    \end{tikzpicture}\label{tik1}
            \begin{tikzpicture}

        \def\shitw{1}
        \def\shith{0.5}

      \draw[thick,->] (0,0) -- (5.5,0) node[anchor=north west] {};
       \draw[] (5.2,-0.5) node[] {$\rho(G_2)$};
      \draw[thick,->] (0,0) -- (0,4.5) node[anchor=south east] {$\rho(G_3)$};

    \foreach \i\x/\w/\y/\h in {
        1/0.2/1/0.75/1/yellow,
        2/0.9/1/2.2/1/black,
        3/2.5/1/0./1/black,
        4/-0.5/1/0.3/1/yellow,
        5/1.75/1/1.5/1/yellow
    }
        {
        \draw[thick] (\shitw+\x ,2pt) -- (\shitw+\x ,-2pt) node[anchor=north] {};
        \draw[thick] (2pt,\shith+\y) -- (-2pt,\shith+\y) node[anchor=east] {};
        \draw[thick] (\shitw+\x+\w , 2pt) -- (\shitw+\x+\w ,-2pt) node[anchor=north] {};
        \draw[thick] (2pt,\shith+ \y+\h) -- (-2pt,\shith+\y+\h ) node[anchor=east] {};
        \draw[loosely dotted] (-\i/5-0.5,\shith+\y) -- (\shitw+\x,\shith+\y) node[anchor=south east] {};
        \draw[loosely dotted] (\shitw+\x,\shith+-\i/5-0.8) -- (\shitw+\x,\shith+\y) node[anchor=south east] {};
        \draw[loosely dotted] (-\i/5-0.5,\shith+\y+\h) -- (\shitw+\x+\w,\shith+\y+\h) node[anchor=south east] {};
        \draw[loosely dotted] (\shitw+\x+\w,\shith+-\i/5-0.8) -- (\shitw+\x+\w,\shith+\y+\h) node[anchor=south east] {};
        \draw (\shitw+\x ,\shith+\y) rectangle (\shitw+\x+\w,\shith+\y+\h) {};
        \draw[] (\shitw+\x+\w/2 ,\shith+\y+\h/2) node[] {$v_{\i}$};
        \draw[|-|] (\shitw+\x, \shith+-\i/5-0.8 ) -- (\shitw+\x+\w, \shith+-\i/5-0.8 ) node[anchor=north] {};
        \draw[|-|] (-\i/5-0.5,\shith+ \y ) -- (-\i/5-0.5, \shith+\y+\h ) node[anchor=north  ] {};

        } {}

    \end{tikzpicture}\label{tiks2}
 }
      
\caption{The $\Delta$-conflict graph from \cref{fig:deltaclique}, geometrically represented as the union of two square intersection graphs. There exist an edge between any $v,u \in V$ if and only if there exist a rectangle of $v$ and a rectangle of $u$ which intersect,~$\rho(v)\bigcap \rho(u) \neq \varnothing$.
}
\label{fig:tsquarerepIS}
\end{figure}

\begin{corollary}
\label{cor:conflict}%
The $\Delta$-conflict graph of any temporal interval graph is a~$(\tau-\Delta+1)$-track $\Delta$-dimensional hyperrectangle graph.
\end{corollary}

\section{\TC}
\label{sec:clique}

In this section we present our results for \TC on temporal interval and unit interval graphs. Recall that \TC generalizes \textsc{Clique}, and is therefore \NP-hard on general graphs, and so there are plausible special cases of \TC on temporal interval graphs. We therefore begin by exploring tractable cases of this problem, and then proceed to describe some hard basic cases.

\subsection{Algorithms}

The first tractable case of \TC on temporal interval graphs is due to a result by K{\"o}nig~\cite{konig2010} who showed that \textsc{Clique} is polynomial time solvable on $2$-track graphs. In our terms, this result can be stated as follows:
\begin{proposition}[{\cite{konig2010}}]
\label{thm:delta2-p}
\TC on temporal interval graphs is polynomial-time solvable when~$\Delta = \tau = 2$.
\end{proposition}

The second tractability result for \TC on temrpoal interval graphs is due Rosgen and Stewart~\cite{RosgenS07}. They present a polynomial time algorithm for intersection graphs of axis parallel rectangles in a fixed dimension. Since the 1-association graph of a temporal interval graph is an intersection graph of axis-parallel hyperrectangles, this result yields the following:
\begin{proposition}[\cite{RosgenS07}]
\label{thm:boxicity-xp}
\TC on temporal interval graphs with~$\Delta = 1$ is solvable in~$O(n^\tau)$ time.
\end{proposition}

We complement both results above by presenting a linear time algorithm for \TC on unit interval graphs when $\Delta, \tau, k = O(1)$. This algorithm can also be viewed as an \FPT-time algorithm with respect to parameter $\tau+k$. To obtain this result, we exploit the geometric properties of the association graphs which were discussed in \cref{sec:Preliminaries}.

\begin{theorem}
\label{thm:C-PIG-FPT}
\TC on temporal unit interval graphs with is solvable in~$\bigO(2^{(k-2)(2\Delta)^{\tau-\Delta+1}}n)$ time.
\end{theorem}
\begin{proof}
Let $\TG$ be an input temporal unit interval graph for \TC.  As discussed in Section~\ref{sec:Preliminaries}, the $\Delta$-association graph $G$ of $\TG$ can be represented as an intersection of $\Delta^{\tau-\Delta+1}$ hypercubes of dimension $(\tau-\Delta+1)$, and these hypercubes can be computed in~$O(\Delta n)$ time, given the unit interval representation of $\TG$.

Clearly, if any corner of a hypercube is included in at least $k-1$ hypercubes, the instance is a yes-instance, since all these hypercubes form a clique of size at least $k$ in $G$. Therefore, let us assume that no corner of any hypercube intersects more than $k-2$ hypercubes. Observe that in this case, the maximum degree of a vertex in $G$ is $(k-2)(2\Delta)^{\tau-\Delta+1}$, as every vertex~$v$ is associated with a set of $\Delta^{\tau-\Delta+1}$ hypercubes of dimension $(\tau-\Delta+1)$ which has $2^{\tau-\Delta+1}$ corners, and every neighbor of $v$ has an associated hypercube that contains some corner of a hypercube associated with $v$. Since \textsc{Clique} on graphs of maximum degree $\delta$ can be solved in $O(2^\delta n)$ time, we get the running time stated in the theorem.
\end{proof}

We next consider approximation algorithms for \TC on temporal interval graphs. Butman \emph{et al.}~\cite{butman2010optimization} give an approximation algorithm for \textsc{Clique} on $t$-interval graphs, a graph class that contains~{$t$-track} graphs. This result directly carries over directly to \TC.
\begin{proposition}[\cite{butman2010optimization}]
\label{thm:delta3}
\TC on temporal interval graphs can be approximated within a factor of $(\Delta^2-\Delta+1)/2$ whenever~$\Delta = \tau$.
\end{proposition}
We complement the above algorithm by proving the following:
\begin{theorem}
\TC on temporal unit interval graphs can be approximated within a factor of~$(2\Delta)^{\tau-\Delta+1}$.
\end{theorem}
\begin{proof}
Let $\TG$ be a temporal unit interval graph given as input to \TC. Our algorithm exploits the the geometric representation of the $\Delta$-association graph $G$ of~$\TG$. Specifically, the algorithm iterates over all vertices, and picks for each vertex a candidate solution by taking the largest clique formed on any corner of any hypercube associated with the vertex. The largest candidate solution is then returned as output. We argue that this algorithm has an approximation ratio of~$(2\Delta)^{\tau-\Delta+1}$.

Consider some vertex $v$ of $G$ which included in a maximum clique of $G$, and let $k$ denote the size of the candidate solution of $v$. Recall that~$v$ is associated with a set of $\Delta^{\tau-\Delta+1}$ hypercubes, each having~$2^{\tau-\Delta+1}$ corners. Since any neighbor of $v$ has an associated hypercube that contains some corner of a hypercube associated with $v$, the maximum size clique of~$G$ (which  by assumption includes $v$) has size less then $k \cdot (2\Delta)^{\tau-\Delta+1}$. Since our algorithm returns a solution of size~$k$, the approximation ratio follows. 
\end{proof}

\subsection{Hardness results}

We next consider intractable cases of \TC on temporal interval graphs. Francis \emph{et al.}~\cite{FrancisGO15} extended the hardness result of Butman \emph{et al.}~\cite{butman2010optimization} and showed that \textsc{Clique} is \NP-hard on $3$-track unit interval graphs. In our terms, this can be stated as follows:
\begin{proposition}[\cite{FrancisGO15}]
\label{thm:tuigdelta3np}
\TC on temporal unit interval graphs is~\NP-hard even if~$\Delta = \tau = 3$.
\end{proposition}

Note that this result should be compared to~\cref{thm:delta2-p} which states that the problem is polynomial-time solvable when~$\Delta = \tau = 2$. However, the case of $\Delta=2$ and large $\tau$ remains open by these two results. We partially close this gap by proving the following: 
\begin{theorem}
\label{thm:delta2-np}
\TC on temporal unit interval graphs is~\NP-hard, even if~$\Delta = 2$ and~$\tau = 11$.
\end{theorem}

For the proof of Theorem~\ref{thm:delta2-np}, we need the following lemma.
\begin{lemma}
\label{lem:delta2-matching}
Let $G=(V,E)$ be such that $M:=\binom{V}{2} \setminus E$ is a matching. Then $G$ is an edge union of 2 unit interval graphs (i.e., a 2-track unit interval graph).
\end{lemma}

\begin{proof}[Proof of \cref{lem:delta2-matching}]
Let $M = \{\{a_1,b_1\}, \ldots, \{a_r,b_r\}\}$ and $V \setminus \bigcup M =\{c_1, \ldots, c_q\}$. The representation is obtained by taking for each $i \in [r]$ the intervals $I^1_{a_i}=(i,i+n)$, $I^2_{a_i}=(i-n-1,i-1)$, $I^1_{b_i}=(i-n-1,i-1)$, $I^2_{b_i}=(i,i+n)$ and for each $i \in [q]$ the intervals~${I^1_{c_i}=(0,n)}$ and $I^2_{c_i}=(0,n)$. Obviously all the intervals are of the same length, it remains to show that they indeed represent $G$.

As $r \le \frac{n}{2}$, the start point of the interval~$I^1_{a_i}=(i,i+n)$ as well as the endpoint of the interval~$I^1_{b_i}=(i-n-1,i-1)$ are contained in the interval $I^1_{c_j}=(0,n)$ and, hence, both~$a_i$ and~$b_i$ are connected to $c_j$ for every $i \in [r]$ and $j \in [q]$. Moreover, all $c_j$'s share the same interval (in both graphs), hence they form a clique as required. If $i < i'$ then the interval $I^1_{a_i}=(i,i+n)$ intersects both the interval $I^1_{a_{i'}}=(i',i'+n)$ and the interval~$I^1_{b_{i'}}=(i'-n-1,i'-1)$ and, hence, $a_i$ is connected to both $a_{i'}$ and $b_{i'}$. Next, if $i > i'$, then the interval $I^2_{b_{i'}}=(i',i'+n)$ intersects both the interval $I^2_{b_{i}}=(i,i+n)$ and the interval $I^2_{a_{i}}=(i-n-1,i-1)$ and, hence, $b_{i'}$ is connected to both $a_{i}$ and $b_{i}$. It follows that $\{a_1, \ldots, a_r\}$ and $\{b_1, \ldots, b_r\}$ are cliques in the union of the interval graphs and $a_i$ is connected to $b_{i'}$ if $i \neq i'$. Finally, the interval $I^1_{a_i}=(i,i+n)$ does not intersect the interval $I^1_{b_i}=(i-n-1,i-1)$ and the interval~$I^2_{a_i}=(i-n-1,i-1)$ does not intersect the interval $I^2_{b_i}=(i,i+n)$. Hence $a_i$ and~$b_i$ are not adjacent for every $i \in [r]$ as required, finishing the proof.
\end{proof}

\begin{proof}[Proof of \cref{thm:delta2-np}]
We reduce the problem \textsc{Independent Set in Cubic Graphs}~\cite{fleischner2010maximum}.
Let~${(G=(V,E), k)}$, where $G$ is a cubic graph and $k$ is an integer, be an instance of \textsc{Independent Set in Cubic Graphs}.
Graph $G$ can be edge colored by 4 colors; let $E =F_1\uplus F_2 \uplus F_3 \uplus F_4$ be a partition of the edges corresponding to such a coloring.
Note that each $F_i$ is a matching.
By \cref{lem:delta2-matching}, graph $(V, \binom{V}{2} \setminus F_i)$ can be represented as an edge union of 2 unit interval graphs.
Let us denote these 2 graphs as $(V,E_{3i-2})$ and $(V,E_{3i-1})$.
Let $E_{3i}=\binom{V}{2}$ for every~$i \in [3]$.

We claim that $((V,E_1, \ldots, E_{11}),k,2)$ is a yes-instance of \TC if and only if $(G, k)$ is a yes-instance of \textsc{Independent Set in Cubic Graphs}.
For the ``if'' direction, let $S$ be an independent set of size $k$ in $G$.
Let $t \in [10]$, we should show that~$S$ forms a clique in $(V,E_t \cup E_{t+1})$. This is clear if $\{t,t+1\} \cap \{3,6,9\} \neq \emptyset$.
Thus, let us assume that~$t=3i-2$ for some $i \in [4]$.
Since $S$ is independent in $G=(V,E)$, it is a clique in~$(V, \binom{V}{2} \setminus E)$ and, as~$F_i \subseteq E$, also in $(V, \binom{V}{2} \setminus F_i)$ which equals $(V,E_{3i-2} \cup E_{3i-1})$ by the construction. 

For the ``only if'' part assume that $S$ is a clique of size $k$ in $(V,E_t \cup E_{t+1})$ for every~$t \in [10]$. We claim that $S$ is an independent set in $G$.
Suppose not and let $e \in E$ have both endpoints in $S$. Then there is an $i \in [4]$ such that $e \in F_i$. But then the edge $e$ is not contained in~$(V,E_{3i-2} \cup E_{3i-1})$ contradicting that $S$ is a clique in this graph.
This concludes our proof, and so the theorem holds.
\end{proof}

We continue to the special case of \TC on temporal interval graphs when $\Delta=1$. In this case our association graph is an intersection graph of $\tau$-dimensional hyperrectangles. As any graph can be represented with a hyperrectangles intersection model provided its dimensionality is high enough, 
\textsc{Temporal 1 Clique} is \NP-hard on temporal interval graphs for sufficiently large~$\tau$. In the following we show that the problem is \Wone-hard when parameterized by $\tau +k$. This complements nicely both Proposition~\ref{thm:boxicity-xp} and Theorem~\ref{thm:C-PIG-FPT}, as it shows that one most likely cannot remove~$\tau$ from the exponent in the running time of Proposition~\ref{thm:boxicity-xp}, nor the unit restriction from Theorem~\ref{thm:C-PIG-FPT}.

\begin{theorem}
\label{thm:boxicity-w1}
\TC on temporal interval graphs is \NP-hard and \Wone-hard with respect to~$\tau+k$, even if~$\Delta = 1$.
\end{theorem}

\begin{proof}
We provide a parameterized reduction from the \textsc{Multicolored Clique} problem~\cite{fellows2009multipleinterval}. Let~$(G,k,c)$, where $G=(V,E)$ is a graph, $k$ is a positive integer  and $c:V \to [k]$ is a (not necessarily proper) coloring of the vertices, be an instance of \textsc{Multicolored Clique}.
We assume without loss of generality, that there are no edges between vertices of the same color. For $1 \le i < j \le k$, let $E_{i} = \{\{u,v\} \in E \mid c(u)=i \vee c(v)=i\}$ and $E^-_i =E\setminus E_i$. We construct the temporal graph $\mathcal{G}'=(V',E'_1, \ldots E'_k)$ as follows. First, we let $V'=E$.
Then, for every $i \in [k]$, we start by letting $E'_i=\{\{e,f\} \mid e \in E^-_i, f \in E\}$. Then for every $v \in V$ we add to $E'_{c(v)}$ the edges $\{\{e,f\} \mid e,f \in E, e \cap f = \{v\}\}$.

We claim that each $(V',E'_i)$ is an interval graph. Indeed, let the vertices of color~${i \in [k]}$ be numbered $\{v \in V \mid c(v)=i\}=\{v^i_1, \ldots v^i_{r_i}\}$. Then the graph $(V',E'_i)$ can be represented by assigning to each $e \in E^-_i$ the interval $(1,2r_i)$ and to each $e \in E_i$ assigning the interval~${(2q-1,2q)}$, where $v^i_q \in e$. Note that this is well defined, since each edge in $E_i$ contains exactly one vertex with color $i$.

We claim that $(\mathcal{G}'=(V',E'_1, \ldots E'_k), \binom{k}{2}, 1)$ is a yes-instance of \TC if and only if $(G,k,c)$ is a yes-instance of \textsc{Multicolored Clique}. For the ``if'' direction let~$S = \{v_1, \ldots, v_k\}$ be a multicolored clique in $G$ such that $c(v_i)=i$ for every $i \in [k]$. The set~$C = \{ \{v_i,v_j\} \mid 1 \le i < j \le k\} \subseteq V'$ has size $\binom{k}{2}$. Therefore we only have to show that it forms a clique in every $(V',E'_i)$, $i \in [k]$. Let $i \in [k]$. The edges in~${C \cap E^-_i}$ are adjacent to all other edges by construction, while edges in~${C \setminus E^-_i = C \cap E_i}$ are adjacent to each other since they all contain the vertex $v_i$, finishing this implication.

For the ``only if'' direction let us assume that there is a set $C$ of size $\binom{k}{2}$ which is a clique in every $(V',E'_i)$, $i \in [k]$. For every $i \in [k]$, since the edges in $C \cap E_i$ form a clique, every two edges in $C \cap E_i$ must share a vertex of color $i$. Since each edge contains at most one vertex of color $i$, there is at most one vertex of color $i$ contained in the edges of $C$. As $|C|=\binom{k}{2}$, and there are $k$ colors, it follows that there is exactly one vertex of each color contained in the edges of $C$ and every two such vertices are connected by an edge in $G$. This concludes our proof, and so the theorem holds.
\end{proof}

Note that \cref{thm:boxicity-w1} is equivalent to the statement that \textsc{Clique} is \Wone-hard when parameterized by the solution size and the  \emph{boxicity} of the input graph; that is, the minimum dimension of boxes which can be used for an intersection representation of a graph.

\section{\tis}
\label{sec:IS}

In this section we consider the \TIS problem restricted to temporal interval graphs. We begin by presenting an approximate algorithm akin to the one presented in Section~\ref{sec:clique}. We then we proceed to discuss some intractable cases for the problem. 

\subsection{Approximation Algorithms}

The first approximation algorithm we mention is due to Erlbach \emph{et al.}~\cite{ErlebachJS05} who considered the \is problem on intersection graphs of ``disk-like" objects. For \TIS, this result can be stated as follows:
\begin{proposition}[\cite{ErlebachJS05}]
\label{prop:Erlbach}%
\TIS on temporal unit interval graphs a PTAS whenever $\tau=\Delta = O(1)$.
\end{proposition}

We next show that the simple greedy algorithm for \is also performs relatively well when both $\tau$ and $\Delta$ are small (but not necessarily equal). The following lemma helps in showing this. 
\begin{lemma}
\label{lem:tdelta}%
Let $G$ be the conflict graph of \TIS on temporal unit interval graphs. For each vertex $v$ of $G$ it holds that any independent set in the graph induced by $v$ and its neighbors is of size at most $2^\Delta(\tau-\Delta+1)$.
\end{lemma}

\begin{proof}
Consider the $\tau-\Delta+1$-track $\Delta$-dimensional hypercube family representation of~$G$ (see \cref{fig:tsqares}). Each vertex $v$ in $G$ is represented by $\tau-\Delta+1$ hypercubes, each of dimension~$\Delta$. Altogether, all these hypercubes corresponding to $v$ have a total of $2^\Delta(\tau-\Delta+1)$ corners. Any neighbor of $v$ is represented by some $\tau-\Delta+1$ hypercubes, one of which must include some corner of a hypercube of~$v$. Moreover, no two independent (\emph{i.e.} non-adjacent) neighbors of $v$ can include the same corner. The lemma thus follows. 
\end{proof}

\begin{figure}[H]
\centering
\begin{tikzpicture}

        \def\x{0}
        \def\y{1}
        \def\size{1}
        \def\push{0.2}

        \draw[line width=0.5mm,] (\x-\size/2 ,\y-\size/2) rectangle (\x+\size/2,\y+\size/2) {};
        \draw[] (\x,\y) node[] {$v^1_i$};

        \draw (\x+\push, \y+\push) rectangle (\push+\x+\size,\push+\y+\size) {};
        \draw (\x-\push, \y+\push) rectangle (-\push+\x-\size,\push+\y+\size) {};
        \draw (\x+\push, \y-\push) rectangle (\push+\x+\size,-\push+\y-\size) {};
        \draw (\x-\push, \y-\push) rectangle (-\push+\x-\size,-\push+\y-\size) {};

        \def\x{1.5}
        \def\y{0}
        \node at (\x+\size/2,\y+\size/2) {\textbf{\ldots}};

        \def\x{4}
        \def\y{-0.5}
        \def\size{1}
        \def\push{0.3}

        \draw[line width=0.5mm,] (\x-\size/2 ,\y-\size/2) rectangle (\x+\size/2,\y+\size/2) {};
        \draw[] (\x,\y) node[] {$v^t_i$};
        \draw (\x+\push, \y+\push) rectangle (\push+\x+\size,\push+\y+\size) {};
        \draw (\x-\push, \y+\push) rectangle (-\push+\x-\size,\push+\y+\size) {};
        \draw (\x+\push, \y-\push) rectangle (\push+\x+\size,-\push+\y-\size) {};
        \draw (\x-\push, \y-\push) rectangle (-\push+\x-\size,-\push+\y-\size) {};

        \def\x{5.5}
        \def\y{0}
        \node at (\x+\size/2,\y+\size/2) {\textbf{\ldots}};

        \def\x{8}
        \def\y{1}
        \def\size{1}
        \def\push{0.2}

        \draw[line width=0.5mm,] (\x-\size/2 ,\y-\size/2) rectangle (\x+\size/2,\y+\size/2) {};
        \draw[] (\x,\y) node[] {$v^\tau_i$};
        \draw (\x+\push, \y+\push) rectangle (\push+\x+\size,\push+\y+\size) {};
        \draw (\x-\push, \y+\push) rectangle (-\push+\x-\size,\push+\y+\size) {};
        \draw (\x+\push, \y-\push) rectangle (\push+\x+\size,-\push+\y-\size) {};
        \draw (\x-\push, \y-\push) rectangle (-\push+\x-\size,-\push+\y-\size) {};

\end{tikzpicture}\label{}

\caption{A set of $\tau$ axis parallel $\Delta$-hypercubes cannot be intersected by more than~$2^\Delta \tau$ disjoint axis parallel $\Delta$ hypercubes.}
\label{fig:tsqares}
\end{figure}
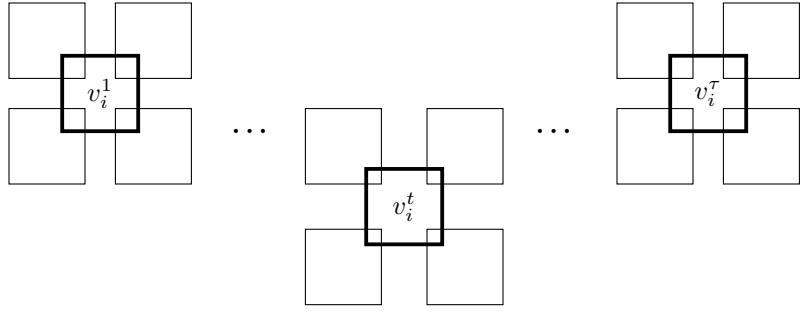

\begin{theorem}\label{thm:cfapx}
\TIS on temporal unit interval graphs can be approximated within a factor of $(\tau-\Delta+1)\cdot 2^\Delta$ in linear time .
\end{theorem}
\begin{proof}
Let $\TG$ be the input temporal unit interval graph of \TIS, and let $G$ be the $\Delta$-conflict graph of $G$. Our algorithm is simply the greedy algorithm the picks an arbitrary vertex of $G$ into the solution, and then removes all its neighbors from the graph. Clearly, the solution returned by this algorithm is an independent set. Moreover, in each step we add one vertex to our solution while removing at most $(\tau-\Delta+1)\cdot 2^\Delta$ vertices of the optimal solution, according to Lemma~\ref{lem:tdelta}. The theorem thus follows. 
\end{proof}

\subsection{Hardness results}

We next describe some intractable cases for \TIS on temporal unit interval graphs. We begin with the case that~$\Delta=1$. Here the conflict graph is simply the union of all layers of $\TG$. Thus, as mentioned in Section~\ref{sec:intro}, the class of all possible 1-conflict graphs is precisely the class of $\tau$-track unit interval graphs. We therefore directly get the following hardness result from the known hardness results for \is in $2$-track unit interval graphs~\cite{bar2006scheduling,fellows2009multipleinterval,jiang2010parameterized}.

\begin{proposition}[\cite{bar2006scheduling,fellows2009multipleinterval,jiang2010parameterized}]
\label{thm:taudelta2}
\TIS on temporal unit interval graphs is \NP-hard, \APX-hard, and \Wone-hard \wrt the solution size $k$ for $\tau \geq 2$ and $\Delta=1$.
\end{proposition}

Next we consider the case where $\Delta=\tau$, by which the class of all $\Delta$-conflict graphs is a subset of the class of $\tau$-dimensional axis-parallel hypercube (intersection) graphs. Marx~\cite{marx2005efficient} showed that \is is \Wone-hard for dimension 2 when parameterized by the solution size $k$. For our setting, this result can be stated as follows.
\begin{proposition}[\cite{marx2005efficient}]
\label{thm:taudelta3}
\TIS on temporal unit interval graphs  is \NP-hard and \Wone-hard \wrt the solution size $k$ for~$\tau=\Delta\geq2$.
\end{proposition}
Note however, that the problem admits a PTAS in this case, see Proposition~\ref{prop:Erlbach}.

\section{Order-Preserving Temporal Interval Graphs}\label{chap:opTG}

In this section, we investigate the computational complexity of \TC and \TIS on so-called order-preserving temporal interval graphs~\cite{Nie2020temporal}.
In \cref{sec:OPPropeties}, we show that both problems can be solved in linear-time on order-preserving temporal graphs. In \cref{sec:opvd}, we show how to solve our problems on non-order-preserving temporal interval graphs via a ``distance-to-triviality'' parameterization~\cite{GHN04}. To this end, we also give an \FPT\ algorithm to compute a minimum vertex deletion set to order preservation with the set size as a parameter. Finally in \cref{sec:opvdhardness}, we show that computing a minimum vertex deletion set to order preservation is \NP-hard.

\subsection{Linear-Time Algorithms assuming Order Preservation}
\label{sec:OPPropeties}%

We say an interval graph \emph{agrees on} or is \emph{compatible with} a total order if has an interval intersection model where the right endpoints of the intervals agree with the total order.
Formally, an interval graph $G$ agrees on~\totalv~if there exists an interval representation $\rho$ for~$G$ such that for every two vertices $\forall v,u \in V$ whose ranking fulfills~$v$ \totalv~$u$, it holds that right$_{\rho}(v)<$~right$_{\rho}(u)$. We call such ordering \emph{right-endpoints (RE) orderings}.
Clearly, any right-endpoints ordering is also a \emph{left-endpoints} ordering of the mirrored intersection model.

\begin{definition}\label{def:orderpreserving}
   A temporal interval graph is \emph{order-preserving} if all of its layers agree on a single RE ordering.
\end{definition}

 Order-preserving temporal unit interval graphs can be recognized in linear time and a corresponding vertex ordering can be computed in linear time as well~\cite{Nie2020temporal}.
The computational complexity of recognizing order-preserving temporal \emph{interval} graphs remains open.

In the following, we show that RE orderings are preserved under both intersection and union of interval graphs.
This means that the conflict graph of an RE order-preserving temporal graph is an interval graph that as well agrees on the RE ordering.
We demonstrate this claim for interval graphs in \cref{thm:opintersection,thm:opunion}.
We start with showing that the intersection of two interval graphs that agree on an RE ordering is again an interval graph that agrees on the ordering.

\begin{lemma}\label{thm:opintersection}
Let $G_1$ and $G_2$ be interval graphs that agree on the total ordering \totalv. Then~$G_1\cap G_2$ is an interval graph that agrees on \totalv.
\end{lemma}

\begin{proof}

Given two interval graphs which agree on an RE ordering,
we can normalize their representations such that for each vertex, the right endpoints in both representations are the same.
To compute an intersection model for their intersection graph,
we can map each vertex to the intersection of their intervals in both representations.
We then show that this mapping is an interval representation of the edge intersection graph.

Let $V$ be a vertex set of size $n$ and let \totalv be a total ordering on $V$ such that for all~$i,j \in [n] $ it holds $ j < i\Leftrightarrow v_j$\totalv~$v_i$.
Since both $G_1$ and $G_2$ agree on \totalv, we can normalize their interval representations so that the right endpoint of the interval associated with each vertex lies on a natural number between $1$ and $n$ according to its ordinal position in~\totalv, formally $\text{index}_{<_V}(v_i)=i$.
Alternatively, we can say that for an interval graph~$G_t$ an interval representation~$\rho_t$ exists such that each $v\in V$ is mapped to an interval of the form~$\rho_t(v) = [a_v, \text{index}_{<_V}(v)]$ with $a_v \in \RR$.
In this normalized representation, the left endpoint of the interval lies on the real line between two natural numbers and is by definition smaller than the right endpoint, that is, $a_v < \text{index}_{<_V}(v)$.

Let $\rho$ be a mapping from the vertex set $V$ to a set of points on $\RR$ such that it holds~${\rho(v)=\rho_1(v)\cap\rho_2(v)}$.
To show that $\rho$ is an interval representation of $G_1\cap G_2$ we first show that $\rho(v)$ is a continuous interval for any $v\in V$,
and that for any two vertices~$v,u\in V$ it holds that~$\rho(v)\cap \rho(u) \neq \varnothing \Leftrightarrow \rho_1(v)\cap \rho_1(u)  \neq \varnothing \wedge \rho_2(v)\cap \rho_2(u)  \neq \varnothing$.

By definition $\rho_1(v)$ and $\rho_2(v)$ are both intervals on the real line,
they are therefore convex sets.
As the mapping $\rho(v)$ is an intersection of two convex sets, it must as well be a convex set and therefore it is interval on the real line.

Let $v_j <_V v_i$, we show that provided $\rho(v_i)\cap \rho(v_j) \neq \varnothing $ then both~${\rho_1(v_i)\cap \rho_1(v_j)  \neq \varnothing}$ and~${\rho_2(v_i)\cap \rho_2(v_j)  \neq \varnothing}$ must hold.
As the interval representations $\rho_1$ and $\rho_2$ are both normalized,
it immediately follows that~$\text{right}_{\rho_1}(v_i)=\text{right}_{\rho_2}(v_i)=i$.
Since both are closed intervals we know that either~$\rho_1(v)\subseteq \rho_2(v)$ or $\rho_2(v)\subseteq \rho_1(v)$.
Without loss of generality, let $\rho_1(v)\subseteq \rho_2(v)$;
it follows that $\rho(v)=\rho_1(v)$.
This means that if~$\rho_2(v_i)\cap \rho_2(v_j)=\varnothing$, then also~$\rho_1(v_i)\cap \rho_1(v_j)=\varnothing$ and therefore $\rho(v_i)\cap \rho(v_j)=\varnothing$.
Regardless, it must hold that~$j\in \rho_1(v_j)\cap\rho_2(v_j)$ as both interval representations of $v_j$ have~$j$ as the right endpoint; it follows that $j\in \rho(v_j)$.
This shows that if $j\in\rho(v_i)$ then~$j\in\rho_1(v_i)$ and $j\in\rho_2(v_i)$.
Therefore, for any $v_j <_V v_i$, if
$\rho(v_i)\cap \rho(v_j) \neq \varnothing $, then both~$\rho_1(v_i)\cap \rho_1(v_j)  \neq \varnothing$ and~$\rho_2(v_i)\cap \rho_2(v_j)  \neq \varnothing$.

Suppose that $\rho(v_i)\cap \rho(v_j) = \varnothing$, but both~$\rho_1(v_i)\cap \rho_1(v_j) \neq \varnothing$ or~${\rho_2(v_i)\cap \rho_2(v_j) \neq \varnothing}$.
This contradicts that $v_j <_V v_i$ because if $\rho_1(v_i)$ contains any point $a\in\rho_1(v_j)$ with $a<j$, then it must contain also $j$ because $\rho_1(v_i)$ is convex.

We have therefore an interval representation $\rho$ which represents the graph $G_1\cap G_2$ because~$\rho(v)\cap\rho(u)\Leftrightarrow \{v,u\}\in E_1 \wedge \{v,u\}\in E_2$ for any $v,u\in V$. Notice that $G_1\cap G_2$ agrees on $<_V$ because $\text{right}_\rho(v_i)=i$.
\end{proof}

Next, we show that the union of two interval graphs agreeing on an RE ordering yields an interval graph that also agrees on the ordering.

\begin{lemma}
\label{thm:opunion}
Let $G_1$ and $G_2$ be interval graphs that agree on the total ordering \totalv.
The union $G=G_1\cup G_2$ is an interval graph that agrees on \totalv.
\end{lemma}
\begin{proof}

The main concept of the proof is analogous to the one for \cref{thm:opintersection}.
Given two interval graphs which agree on an RE ordering,
we can normalize their representations such that for each vertex, the right endpoints in both representations are the same.
To compute an intersection model for their \emph{union} graph,
we can map each vertex to the \emph{union} of their intervals in both representations.
We then show that this mapping is an interval representation of the edge-\emph{union} graph.

Let $\rho$ be a mapping from the vertex set $V$ to a set of points on $\RR$ such that it holds~${\rho(v)=\rho_1(v)\cup\rho_2(v)}$.
To show that $\rho$ is an interval representation of $G_1\cup G_2$ we first show that $\rho(v)$ is a continuous interval for any $v\in V$,
and that for any two vertices~$v,u\in V$ it holds that $\rho(v)\cap \rho(u) \neq \varnothing \Leftrightarrow \rho_1(v)\cap \rho_1(u)  \neq \varnothing \vee \rho_2(v)\cap \rho_2(u)  \neq \varnothing$.

By definition $\rho_1(v)$ and $\rho_2(v)$ are both closed and normalized intervals on the real line such that ${\text{right}_{\rho_1}(v_i)=\text{right}_{\rho_2}(v_i)=i}$.
As we observed in \cref{thm:opintersection}, since both intervals have the same right endpoint it holds that either~$\rho_1(v)\subseteq \rho_2(v)$ or $\rho_2(v)\subseteq \rho_1(v)$.
Without loss of generality, assume that~${\rho_2(v)\subseteq \rho_1(v)}$, it follows that $\rho(v)=\rho_1(v)=\rho_1(v)\cup\rho_2(v)$.

Suppose that $\rho(v_i)\cap \rho(v_j) \neq \varnothing$, but both~$\rho_1(v_i)\cap \rho_1(v_j)=\varnothing$ or~${\rho_2(v_i)\cap \rho_2(v_j) = \varnothing}$.
This contradicts that $v_j <_V v_i$ because if $\rho(v_i)$ contains any point $a\in\rho(v_j)$ with $a<j$, then it must contain also $j$ because $\rho(v_i)$ is convex.
If $j \in \rho(v_i)$ then trivially $j\in
\rho_1(v_i)$.

If $\rho(v_i)\cap \rho(v_j) = \varnothing$ but either~$\rho_1(v_i)\cap \rho_1(v_j) \neq\varnothing$ or~${\rho_2(v_i)\cap \rho_2(v_j)  \neq \varnothing}$, then it contradicts the fact that either~$\rho_1(v_i)\subseteq\rho_2(v_i)=\rho(v_i)$ or that~$\rho_2(v_i)\subseteq\rho_1(v_i)=\rho(v_i)$.

We have therefore an interval representation $\rho$ which represents the graph $G_1\cup G_2$ because~$\rho(v)\cap\rho(u)\Leftrightarrow \{v,u\}\in E_1 \vee \{v,u\}\in E_2$ for any $v,u\in V$. Notice that $G_1\cup G_2$ agrees on $<_V$ because $\text{right}_\rho(v_i)=i$.
\end{proof}

Using both \cref{thm:opunion,thm:opintersection} we arrive at the following corollary.
\begin{corollary}
\label{cor:opDist}
Let $\TG$ be an order-preserving temporal interval graph, and let $\Delta \geq 1$ be some integer. Then both the $\Delta$-association and the $\Delta$-conflict graph of $\TG$ are interval graphs.
\end{corollary}
Since both \textsc{Clique} and \is can be solved in linear time on interval graphs~\cite{gavril1974intersection,DBLP:journals/siamcomp/RoseTL76}, we get the main result of this subsection.
\begin{theorem}
\label{thm:oppolytime}
\TC and \TIS on order-preserving temporal interval graphs are both solvable in linear time.
\end{theorem}

\subsection{FPT-Algorithms for Vertex Deletion to Order Preservation}\label{sec:opvd}

Now we generalize \cref{thm:oppolytime} and show how to solve \TC and \TIS on \emph{almost} order-preserving temporal unit interval graphs, that is, graphs that most of their vertices agree on a common ordering.
To this end, we define a \emph{distance} of a temporal graph to order preservation.
This distance is measured by the size of the minimum vertex set that obstructs the compatibility of a total RE order of a temporal interval graph. We define it as follows and give an illustration in \cref{fig:opvd}.

\begin{definition}[\opvd]\label{def:opvd}
Let $\TG=(V,\TE,\tau)$ be a temporal interval graph.
A \emph{\OPVD} (\opvd) is a set of vertices~$V' \subseteq V$ such that $\TG-V'$ is  order-preserving.
\end{definition}

\begin{figure}[t]
  \centering
    \subfloat[\centering Graph representations of the two interval graphs]{
    {

          \begin{tikzpicture}[main/.style = {draw,fill=white, circle},scale=1.1]

          \draw[line width=0.7mm]  (4,4) to (5,2); 
          \draw[line width=0.7mm]  (5,2) to (4,0); 
          \draw[line width=0.7mm]  (2,0) to (4,0);
          \draw[line width=0.7mm]  (1,2) to (2,0);
          \draw[line width=0.7mm]  (0.9,2) to (1.9,4);

          \draw[line width=0.3mm]   (4.1,4) to (5.1,2);
        \draw[line width=0.3mm]   (4.1,4) to (2,0);
        \draw[line width=0.3mm]   (4.1,4) to (1,2);
          \draw[line width=0.3mm]   (5.1,2) to (4.1,0);
          \draw[line width=0.3mm]   (0.9,2) to (1.9,0);
          \draw[line width=0.3mm]   (1,2) to (2,4);

          \node[main] (1) at (4,4) {$v_1$};
          \node[main] (2) at (5,2) {$v_2$};
          \node[main] (3) at (4,0) {$v_3$};
          \node[main] (4) at (2,0) {$v_4$};
          \node[main] (5) at (1,2) {$v_5$};
          \node[main] (6) at (2,4) {$v_6$};

        \end{tikzpicture}

     }
    }%
    \qquad
    \subfloat[\centering Two intersection models of $G_1 -\{v_4\}$ and $G_2 -\{v_4\}$ which are compatible with~$<_{V'}$ ]{{

    \begin{tikzpicture}[scale=1.1]

      \draw[thick,->] (0,0) -- (6,0) node[anchor=north west] {};

    \foreach \x/\b/\r/\l in {1/0/0/3, 2/0.5/0.5/2, 3/1.5/1.5/1, 4/3.5/2.5/5, 5/3.5/3.5/6 }
        {
            \draw[line width=0.2mm, -|] (\x,\x/3-0.02) -- (\r,\x/3-0.02) node[anchor=north west] {};
            \draw[line width=0.5mm, -|] (\x,\x/3+0.04) -- (\b,\x/3+0.04) node[anchor= east] {};
            \filldraw (\x , \x/3) circle (2pt) ;
            \node at (\x+0.3 , \x/3-0.1) {$v_{\l}$};
        }

    \end{tikzpicture}

    }}%

    \caption{Two interval graphs, $G_1$ (thick) and $G_2$ (thin), that \emph{do not} have a common RE ordering.
    The vertex subset $\{ {v_4} \}$ is an OPVD set of the temporal graph $\TG = [G_1,G_2]$ as both~$G_1-\{ {v_4} \}$ and~$G_2-\{ {v_4} \}$ agree on $<_{V'} = [v_3,v_2,v_1,v_5,v_6]$.  Two compatible interval representations are illustrated in (b). }%
  \label{fig:opvd}%
\end{figure}
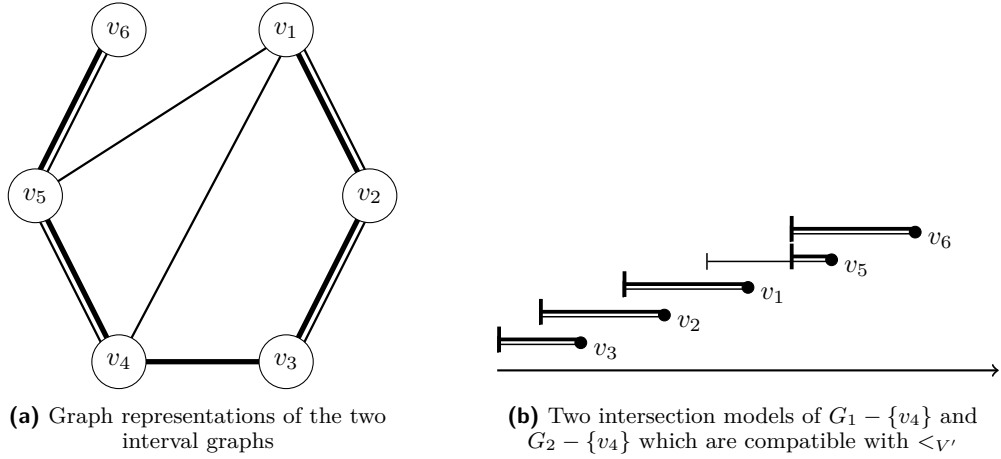

The size of the minimum \opvd set measures how many vertices obstruct a total RE order for a temporal interval graph.
We denote the cardinality of the minimum \opvd by $\ell$.
A brute-force algorithm checks every subset of the vertex set to find a solution to \TC and \TIS.
Given an $\ell$-sized \opvd set we can brute-force the power set of the \opvd (which has size $2^\ell$)
and then check against the rest of the order-preserving graph in polynomial time.

\begin{theorem}
\label{thm:fptopClique}
\TC can be decided in~$2^\ell\cdot n^{\bigO(1)}$ time when given a size-$\ell$ \opvd set of the input temporal graph.
\end{theorem}
\begin{proof}
  The idea is as follows. Given an order-preserving vertex deletion set~$S$ of size~$\ell$,
  we brute-force its power set.
  Let $G_\text{op}$ be the association graph of~$\TG-S$.
  The graph~$G_\text{op}$ is, by definition, an interval graph.
  For each subset $X$ of $S$ we compute in polynomial time~${G_X}$ as
  $ \big( \bigcap\limits_{v\in X} N_G(v) \setminus \{ u \mid \exists w \in X~,~\{u,w\}\notin E \}   \big) \setminus S$. In other words, ${G_X}$ is the neighborhood intersection of all vertices in~$X$ and contains only vertices that are adjacent to all vertices in~$X$ and none of the vertices of $S$, that makes it an interval graph.
  We then find the maximum clique $C$ on~${G_X}$, we do that in linear time since it is an interval graph.
  If~$|X\bigcup C|\geq k$ then we have a yes-instance,
  Otherwise it is a no-instance.
\end{proof}

We next show that the same approach works as well for \TIS on almost order-preserving temporal interval graph.

\begin{theorem}
\label{thm:fptop}
\TIS can be decided in~$2^\ell\cdot n^{\bigO(1)}$ time when given a size-$\ell$ \opvd set of the input temporal graph.
\end{theorem}
\begin{proof}
  The idea is as follows. Given an order-preserving vertex deletion set~$S$ of size $\ell$,
  we brute-force its power set.
  Let $G_\text{op}$ be the conflict graph of~$\TG-S$.
  The graph $G_\text{op}$ is, by definition, an interval graph.
  For each subset $X$ of $S$ we compute~${G_X}$ as ${G_\text{op}-N_{G_\text{op}}(X)}$,
  the neighbors of $X$ from the conflict interval graph.
  As $G_X$ is an interval graph, we can compute a maximum independent set $V'$ of $G_X$ in linear time,
  then check in quadratic time whether $X\cup V'$ is an independent set of size $k$ in the conflict graph of~$\TG$.
  If $X\cup V'$ is an independent set of size at least $k$, then we have a yes-instance.

  Any independent set of size $k$ must clearly be divisible into two subsets, a subset of~$X \subseteq S$ (that includes the trivial subset) and a subset of $V\setminus S$.
  Any independent set on $\TG$ must be also an independent set on the subgraph induced by $V\setminus S$.
  If we exhaust all of the subsets of~$S$ and do not find an independent set of size at least $k-|X|$ on~$\TG-(S \cup N_{G_\text{op}}[X])$ for~$X \subseteq S$,
  then we can conclude that such set does not exist.
  In such case the instance is a no-instance.
  The power set of $S$ is of size $2^\ell$, which means it takes~$2^\ell\cdot n^{\bigO(1)}$ time to exhaust all subsets of~$S$.
\end{proof}

Since the \fpt~algorithm for \TC and \TIS
parameterized by the minimum \opvd $\ell$ behind
\cref{thm:fptop} and \cref{thm:fptopClique} requires access to an $\ell$-sized \opvd set, we present an \fpt-algorithm to compute a minimum \opvd for a given temporal unit interval graph.
We do this by providing a reduction to the so-called \textsc{Consecutive Ones Submatrix by Column Deletions} problem, for which efficient algorithms are known~\cite{dom2010approximation,narayanaswamy2015obtaining}.

 Before we describe the reduction, we give an alternative characterization of order-preserving temporal unit interval graphs. We will use this characterization in our \fpt-algorithm to compute a minimum \opvd.
  As we show in the next lemma, a temporal unit interval graph~$\TG$ is order-preserving if and only if its vertices vs.\ maximal cliques matrix has the so-called \emph{consecutive ones property}\footnote{A 0-1-matrix has the \emph{consecutive ones property} if there exists a permutation of the columns such that in each row all ones appear consecutively.} (C1P). Note that it is known that the vertices vs.\ neighborhoods matrix also has the consecutive ones property in this case~\cite{Nie2020temporal}.

\begin{lemma}
\label{thm:recognizeoppig}
  A temporal unit interval graph is order-preserving if and only if its \emph{vertices vs.\ maximal cliques matrix} has the consecutive ones property.
\end{lemma}
\begin{proof}
  Testing for the consecutive ones property for a matrix can be done in linear time~\cite{tucker1972structure}.
  To test a temporal unit interval graph for order preservation we compute its vertices vs.\ maximal cliques matrix and test it for the consecutive ones property.
  We say that $\TG$'s set of maximal cliques $\mathcal{C}$ is the union of sets of maximal cliques of each layer of $\TG$.
  The vertices vs.\ maximal cliques matrix $M$ is a binary matrix in which~$M_{i,j}=1$ if and only if the vertex~$v_i \in V$ is a member of $C_j \in \mathcal{C}$.
  It is left to show that a temporal unit interval graph~$\TG$ is order-preserving if and only if its vertices vs.\ maximal cliques matrix has the consecutive ones property.

  $(\Rightarrow)$ If $\TG$ is order-preserving, then there exists an ordering $<_V$ such that every layer has an interval representation $\rho$ in which the right endpoints of all intervals agree on $<_V$.
  Let $M$'s columns be ordered by $<_V$.
  If $M$ is not in its petrie form~\footnote{A 0-1-matrix is in its \emph{petrie form} (if it has one) if the columns are permuted in a way such that the ones appear consecutively in all rows.}, then it must mean that there exists a clique $C$ in $\mathcal{C}$ whose members are not consecutive in $<_V$.
  In other words, there exist $u,v,w\in V$ such that $u <_V w <_V v$, for which~$u,v\in C$ and $w \notin C$.
  Since $u$ and $v$ are adjacent and $u <_V v$, we know that left$(v) < \text{right}(u)$.
  We know also that the length of $\rho(v)$ is exactly~1.
  This definitely means that $v$ and $w$ intersect because $\text{right}(w) \in [\text{right}(u), \text{right}(v)]$.
  However since $w \notin C$, $w$ and $u$ cannot be adjacent.
  This is a contradiction since  $\text{right}(v) - \text{right}(u) < 1$ and $\text{right}(w) - \text{left}(w) = 1$.
  If~$\TG$ is order-preserving, then $M$ must have the consecutive ones property.

  $(\Leftarrow)$ If $M$ has the consecutive ones property, then there exists an ordering $<_V$ so that the vertices vs.\ maximal cliques matrix $M_t$ of every layer $G_t \in \TG$ is in its petrie form, when its columns are permuted according to $<_V$.
  Let $\mathscr{I}_t(v)$ be the union of all maximal cliques in layer $G_t$ which contain $v$.
  We know that for every $v \in V$ the vertices of $\mathscr{I}_t(v)$ are consecutive in $<_V$ \cite{booth1976testing}.
  Let $\text{index}(v)$ be the index of $v$ in  $<_V$ and let~${\rho_{G_t}(v) = [\text{min}\{\text{index}(u) \mid u \in \mathscr{I}_t(v)\}-1+\text{index}(v)\cdot\varepsilon , \text{index}(v)  ]}$, for some $0<\varepsilon < 1/|V|$. First, note that no two intervals are contained in each other. This means that there is an equivalent interval representation where all intervals have unit length~\cite{roberts1969indifference}.

  By showing ${\{u,v\}\in E_t \Leftrightarrow \rho_t(v) \cap \rho_t(u) \neq \varnothing}$ we effectively show that $\rho_{G_t}$ is an interval representation of $G_t$.
  If $\{ u,v \} \in E_t$, then there must exist a maximal clique~$C$ so that~$u,v \in C$ and thus $u \in \mathscr{I}_t(v)$ and $v \in \mathscr{I}_t(u)$.
  Assume $u <_V v$, then $\text{left}_{\rho_{G_t}}(v) \leq \text{right}_{\rho_{G_t}}(u)$ and by that~$\rho_{G_t}(u) \cap \rho_{G_t}(v) \neq \varnothing$.
  If $\{ u,v \} \notin E_t$ then there is no clique~$C$ so that $u,v \in C$.
  Assume that $u <_V v$, then~${\rho_{G_t}(u) \cap \rho_{G_t}(v) \neq \varnothing}$ if and only if $\text{left}_{\rho_{G_t}}(v) \leq \text{right}_{\rho_{G_t}}(u)$.
  This cannot be because $\text{left}_{\rho_{G_t}}(v)$ is exactly the right endpoint of $v$'s lowest neighbors in $<_V$. If $\text{right}_{\rho_{G_t}}(u)$ is right of $v$'s lowest neighbor's right endpoint, then the vertices of $\mathscr{I}_t(v)$ are not consecutive in $<_V$.
  This contradicts the fact that $M$ has the consecutive ones property.
\end{proof}

  Since there is a bijection between the columns of the vertices vs.\ maximal cliques matrix~$M$ of $\TG$ and $\TG$'s vertices,
  we can use $M$ as input for the \textsc{Consecutive Ones Submatrix by Column Deletions} problem and apply existing algorithms for that problem~\cite{dom2010approximation,narayanaswamy2015obtaining} to obtain a vertex-maximal temporal subgraph of $\TG$ that is order-preserving. This allows us to obtain the following result.

\begin{theorem}\label{thm:fpt:opvd}
  A minimum \opvd for a given temporal unit interval graph can be computed in $10^\ell n^{\bigO(1)}$ time, where $\ell$ is the size of a minimum \opvd.
\end{theorem}
\begin{proof}

    In \cref{thm:recognizeoppig} we have shown that a temporal unit interval graph is order-preserving if and only if its vertices vs.\ maximal cliques matrix has the consecutive ones property. We provide a reduction to the \textsc{Consecutive Ones Submatrix by Column Deletions} problem.
        Formally, in \textsc{Consecutive Ones Submatrix by Column Deletions} we are given a binary matrix $M \in \{0,1\}^{m\times n}$ and are asked whether there exists a submatrix $M'$ with the consecutive ones property, such that $M'$ is obtained with not more than $\ell$ column deletions from $M$.
    \textsc{Consecutive Ones Submatrix by Column Deletions} is known to be \FPT~with respect to the column deletion set size and it can be decided in~$10^\ell n^{\bigO(1)}$ time~\cite{dom2010approximation,narayanaswamy2015obtaining}.

  Let~$\TG=(V,\TE)$ be a temporal unit interval graph with~$\mathcal{C}$ as maximal cliques set.
  Let $M$ be the vertices vs.\ maximal cliques matrix of $\TG$.
  If $M$ does not have the consecutive ones property,
  then we can find a set of $\ell$ columns in $10^\ell n^{\bigO(1)}$ time so that
  when deleted from~$M$, the resulting matrix $M'$ has the consecutive ones property.
  The columns of $M$ are mapped to vertices of $V$,
  the image of the deleted columns $V'$ is the \opvd set.
  We can find in linear time an ordering $<_V'$ of $M'$ columns such that $M'$ is in its petrie form.
  All layers of the graph $\TG-V'$ agree on $<_V'$.
\end{proof}

This provides us with an efficient algorithm for both \TIS and \TC on ``almost'' ordered temporal unit interval graph.
Namely, find in $10^\ell n^{\bigO(1)}$ time a minimum \opvd set in the input temporal unit interval graph using \cref{thm:fpt:opvd},
then decide in $2^\ell n^{\bigO(1)}$ time if we have a yes-instance of \TC or \TIS using \cref{thm:fptopClique} or \cref{thm:fptop} respectively. Overall, we arrive at the following result.

\begin{corollary}
\label{cor:fptop}

\TC and \TIS can both be decided in~$10^\ell\cdot n^{\bigO(1)}$ time if the input temporal graph is a temporal unit interval graph, where~$\ell$ is the size of a minimum \opvd of the input temporal graph.

\end{corollary}




\subsection{NP-Hardness of Vertex Deletion to Order Preservation}
\label{sec:opvdhardness}

Finally, we show that computing a minimum \opvd for a given temporal unit interval graph is \NP-hard. This complements \cref{thm:fpt:opvd} as it implies that we presumably cannot improve \cref{thm:fpt:opvd} to a polynomial-time algorithm.

\begin{theorem}\label{thm:np:opvdPIG}
 Computing a minimum \opvd for a given temporal unit interval graph is \NP-hard.
\end{theorem}
\begin{proof}
To show \NP-hardness, we present a polynomial time many-one reduction from the
\npcomp~\textsc{Consecutive Ones Submatrix by Column Deletions} problem~\cite{hajiaghayi2002note} to the problem of computing an \opvd of size at most $\ell$ for a given temporal unit interval graph. Note that this implies \NP-hardness of the optimization problem of finding a minimum \opvd.

Formally, in \textsc{Consecutive Ones Submatrix by Column Deletions} we are given a binary matrix $M \in \{0,1\}^{m\times n}$ and are asked whether there exists a submatrix $M'$ with the consecutive ones property, such that $M'$ is obtained with not more than $\ell$ column deletions from $M$. Note that we can assume w.l.o.g.\ that there are at least two ones in each row of $M$, otherwise we can delete the row since its ones are consecutive for all permutations of the columns.

Our reduction works as follows. Given a binary matrix $M \in \{0,1\}^{m\times n}$ with $m$ rows and $n$ columns, we create a temporal graph $\mathcal{G}$ with $n$ vertices $V=\{1,\ldots, n\}$, one for each column, and $m$ layers, one for each row.
In each layer $G_t$ for $1\le t\le m$, we add an edge between vertices~$i$ and~$j$ if $M_{t,i}=1$ and $M_{t,j}=1$. This finished the construction of $\mathcal{G}$, which can clearly be done in polynomial time.

Next, we argue that $\mathcal{G}$ is a temporal unit interval graph. To this end, note that every layer~$G_t$ of $\mathcal{G}$ is a single clique (consisting of vertices $i$ with $M_{t,i}=1$) and some isolated vertices (the vertices $i$ with $M_{t,i}=0$). Hence, we can clearly find a unit interval representation for every layer $G_t$ of $\mathcal{G}$.

To prove the correctness of the reduction, we first observe that $M$ is the vertices vs.\ maximal cliques matrix of $\mathcal{G}$: there is exactly one non-trivial maximal clique in each layer $G_t$ containing the vertices $i$ with $M_{t,i}=1$. We show $\ell$ columns can be deleted from $M$ such that the remaining matrix $M'$ has the consecutive ones property if and only if $\mathcal{G}$ admits an \opvd of size $\ell$.

$(\Rightarrow)$ Assume there are $\ell$ columns that can be deleted from $M$ such that the remaining matrix $M'$ has the consecutive ones property. Then $M'$ corresponds to vertices vs.\ maximal cliques matrix of $\mathcal{G}'$ which is obtained from $\mathcal{G}$ by removing the $\ell$ vertices corresponding to the deleted columns of $M$. By \cref{thm:recognizeoppig} we have that $\mathcal{G}'$ is an order-preserving temporal unit interval graph. It follows that the removed vertices form an \opvd of size $\ell$ for $\mathcal{G}$.

$(\Leftarrow)$ Assume $\mathcal{G}$ admits an \opvd $X$ of size $\ell$. Then let $M'$ be the matrix obtained from~$M$ by deleting the $\ell$ columns corresponding to the vertices in $X$. Now we have that~$M'$ is the vertices vs.\ maximal cliques matrix of $\mathcal{G}-X$, which is an order-preserving temporal unit interval graph. By \cref{thm:recognizeoppig} we have that $M'$ has the consecutive ones property and hence that $(M,\ell)$ is a yes-instance of \textsc{Consecutive Ones Submatrix by Column Deletions}.
\end{proof}

\section{Conclusion}\label{chap:conc}

We study naturally motivated temporal versions of the classic \textsc{Clique} and \is problems, which are called \TC and \TIS, respectively. We introduce the latter whereas the former has been investigated before~\cite{bentert2019listing,himmel2017adapting,viard2016computing}. We focused mostly on the case where all layers of the input temporal graph are unit interval graphs. For these, we presented a number of algorithms, both approximate and exact, and hardness results that attempt at given a broad picture of the computational complexity of both problems. 

As even the most basic cases are hard for these problems, we focused also on the case where the temporal interval graphs are order preserving. We present linear-time algorithms for both \TC and \TIS for this special case, and generalized these algorithms to  \FPT-algorithms for the vertex deletion distance to order preservation parameter. This generalization heavily relies on our result that order preservation is retained under edge-union and edge-intersection, which is of independent interest since it may also be useful in the context of related problem such as \textsc{Temporal Vertex Cover}~\cite{akrida2020temporal} or \textsc{Temporal Coloring}~\cite{mertzios2019sliding}.

An immediate future work direction is to generalize our results for temporal (non-unit) interval graphs. For most of our results it remains open whether they generalize. We believe that our approximation algorithm does not easily adapt. In fact even for two layers it is unclear how to approximate \MTIS. Our \FPT-algorithm for \TC and \TIS parameterized by the vertex deletion distance to order preservation generalizes to the non-unit interval case assuming the deletion set is part of the input. We leave for future research how to efficiently compute a minimum vertex deletion set to order preservation for temporal non-unit interval graphs.




\bibliography{bibfile}



\end{document}